\documentclass[AMA,Times1COL]{WileyNJDv5} %STIX1COL,STIX2COL,STIXSMALL

\articletype{Article Type}%

\received{Date Month Year}
\revised{Date Month Year}
\accepted{Date Month Year}
\journal{Statistics in Medicine}
\volume{44: e70089}
\copyyear{2025}
\startpage{1}
\def\beqr{\begin{eqnarray}}
	\def\eeqr{\end{eqnarray}}
\def\beqrs{\begin{eqnarray*}}
	\def\eeqrs{\end{eqnarray*}}

\def\1{{\bf 1}}

\usepackage{subfigure}
\usepackage{epsfig}
\usepackage{caption}

\begin{document}

\title{An alternative measure for quantifying the heterogeneity in meta-analysis}

\author[1]{Ke Yang}

\author[2]{Enxuan Lin}

\author[3]{Wangli Xu}

\author[4]{Liping Zhu}

\author[5]{Tiejun Tong}

\authormark{YANG \textsc{et al.}}
\titlemark{An alternative measure for quantifying the heterogeneity in meta-analysis}

\address[1]{\orgdiv{Department of Statistics and Data Science}, \orgname{Beijing University of Technology}, \orgaddress{\state{Beijing}, \country{China}}}

\address[2]{\orgdiv{Department of Biostatistics and Information}, \orgname{Innovent Biologics, Inc.}, \orgaddress{\state{Beijing}, \country{China}}}

\address[3]{\orgdiv{School of Statistics}, \orgname{Renmin University of China}, \orgaddress{\state{Beijing}, \country{China}}}

\address[4]{\orgdiv{Institute of Statistics and Big Data}, \orgname{Renmin University of China}, \orgaddress{\state{Beijing}, \country{China}}}

\address[5]{\orgdiv{Department of Mathematics}, \orgname{Hong Kong Baptist University}, \orgaddress{\state{Hong Kong}}}

\corres{Corresponding author Tiejun Tong, Department of Mathematics, Hong Kong Baptist University, Hong Kong. \email{tongt@hkbu.edu.hk}}

%\fundingInfo{Text}
%\JELinfo{ejlje}

\abstract[Abstract]{Quantifying the heterogeneity is an important issue in meta-analysis, and among the existing measures, the $I^2$ statistic is most commonly used. In this paper, we first illustrate with a simple example that the $I^2$ statistic is heavily dependent on the study sample sizes, mainly because it is used to quantify the heterogeneity between the observed effect sizes. To reduce the influence of sample sizes, we introduce an alternative measure that aims to directly measure the heterogeneity between the study populations involved in the meta-analysis. We further propose a new estimator, namely the $I^2_{\rm A}$ statistic, to estimate the newly defined measure of heterogeneity. For practical implementation, the exact formulas of the $I^2_{\rm A}$ statistic are also derived under two common scenarios with the effect size as the mean difference (MD) or the standardized mean difference (SMD). Simulations and real data analysis demonstrate that the $I^2_{\rm A}$ statistic provides an asymptotically unbiased estimator for the absolute heterogeneity between the study populations, and it is also independent of the study sample sizes as expected. To conclude, our newly defined $I_A^2$ statistic can be used as a supplemental measure of heterogeneity to monitor the situations where the study effect sizes are indeed similar with little biological difference. In such scenario, the fixed-effect model can be appropriate; nevertheless, when the sample sizes are sufficiently large, the $I^2$ statistic may still increase to 1 and subsequently suggest the random-effects model for meta-analysis.}


\jnlcitation{\cname{%
\author{Yang K},
\author{Lin E},
\author{Xu W},
\author{Zhu L}, and
\author{Tong T}}.
\ctitle{An alternative measure for quantifying the heterogeneity in meta-analysis.} \cjournal{\it Statistics in Medicine} \cvol{2024;00(00):1--18}.}

\maketitle

\section{Introduction}\label{sec1}
Meta-analysis is a statistical technique for evidence-based practice, which aims to synthesize multiple studies and produce a summary conclusion for the whole body of research \cite{egger1997meta}. In the literature, there are two commonly used statistical models for meta-analysis, namely, the fixed-effect model and the random-effects model. Among them, the fixed-effect model assumes that the effect sizes of different studies are all the same, which is somewhat restrictive and may not be realistic in practice. The effect sizes often differ between the studies due to variability in study design, outcome measurement tools, risk of bias, and the participants, interventions and outcomes studied \cite{higgins2019cochrane}, etc. Such diversity in the effect sizes is known as the heterogeneity. When the heterogeneity exists, the random-effects model ought to be applied for meta-analysis. In such scenarios, it is of great importance to properly quantify the heterogeneity so as to explore the generalizability of the findings from a meta-analysis.

To describe the heterogeneity in detail, we first introduce the random-effects model for meta-analysis. Let $k\ge 2$ be the total number of studies, and $y_i$ be the observed effect sizes from the studies $i=1,\dots,k$. For each study with true effect size $\mu_i$, we assume that $y_i$ is normally distributed with mean $\mu_i=E(y_i|\mu_i)$ and variance $\sigma_{y_i}^2={\rm var}(y_i|\mu_i)$. Moreover, to account for the heterogeneity between the studies, we also assume that the individual effect sizes $\mu_i$ follow another normal distribution with mean $\mu$ and variance $\tau^2>0$. Taken together, the random-effects model for meta-analysis can be expressed as
\beqr\label{meta}
y_i=\mu+\delta_i+\epsilon_i,\quad \delta_i\stackrel{\text{i.i.d.}}{\sim} N(0,\tau^2), \quad \epsilon_i\stackrel{\text{ind}}{\sim} N(0,\sigma_{y_i}^2),
\eeqr
where ``i.i.d." represents independent and identically distributed, ``ind" represents independently distributed, $\tau^2$ is the between-study variance, and $\sigma_{y_i}^2$ are the within-study variances. In addition, the study deviations $\delta_i=\mu_i-\mu$ and the random errors $\epsilon_i$ are assumed to be independent of each other. When $\delta_i$ are all zero, model (\ref{meta}) reduces to the fixed-effect model and there is no heterogeneity between the studies.

To test the existence of heterogeneity for model (\ref{meta}),  Cochran (1954) \cite{cochran1954combination} proposed the $Q$ statistic as
\beqr\label{qs}
Q=\sum_{i=1}^kw_i\left(y_i-\frac{\sum_{i=1}^kw_iy_i}{\sum_{i=1}^kw_i}\right)^2,
\eeqr
where $w_i=1/\sigma_{y_i}^2$ are the inverse-variance weights. Noting that $\sigma_{y_i}^2$ can often be estimated with high precision, it is a common practice in meta-analysis that the within-study variances are regarded as known. Nevertheless, when used as a measure of heterogeneity, it is often criticized that the value of $Q$ will increase with the number of studies. Another measure for heterogeneity is to apply the between-study variance $\tau^2$, yet it is known to be specific to a particular effect metric, making it impossible to compare across different meta-analyses\cite{dersimonian1986meta}. To have a fair comparison,  Higgins and Thompson (2002) \cite{higgins2002quantifying} and  Higgins
et al. (2003) \cite{higgins2003measuring} introduced the $I^2$ statistic by a two-step procedure, under the assumption that the within-study variances $\sigma_{y_i}^2=\sigma_y^2$ are all the same. They first defined the measure of heterogeneity between the studies as
\beqr\label{icch}
{\rm ICC}_{\rm HT}=\frac{\tau^2}{{\rm var}(y_i)}=\frac{\tau^2}{\tau^2+\sigma_y^2},
\eeqr
and then proposed
\beqr\label{i2}
I^2=\frac{\hat\tau^2}{\hat\tau^2+\tilde\sigma_y^2}
=\max\left\{\frac{Q-(k-1)}{Q},0\right\}
\eeqr
to estimate the unknown ${\rm ICC}_{\rm HT}$, where $\hat\tau^2={\mbox{max}}\{\{Q-(k-1)\}/(\sum_{i=1}^kw_i-\sum_{i=1}^kw_i^2/\sum_{i=1}^kw_i),0\}$ is the DerSimonian-Laird estimator \cite{dersimonian1986meta} and $\tilde\sigma_y^2=(k-1)/(\sum_{i=1}^kw_i-\sum_{i=1}^kw_i^2/\sum_{i=1}^kw_i)$. When the within-study variances are all the same, $\tilde\sigma_y^2$ is identical to the common $\sigma_y^2$. Otherwise if they differ,  B{\"o}hning et al. (2017) \cite{bohning2017some} has showed that $\tilde\sigma_y^2$ is asymptotically identical to the harmonic mean $(\sum_{i=1}^kw_i/k)^{-1}$ of the within-study variances. Moreover, the $I^2$ statistic is also guaranteed to be within the interval $[0,1)$, which is appealing in that it does not depend on the number of studies and is irrespective of the effect metric.

Thanks to its nice properties, the $I^2$ statistic is nowadays routinely reported in the forest plots for meta-analyses, and/or used as a criterion for model selection between the fixed-effect model and the random-effects model. In Google Scholar, as of January 2025, the two papers by Higgins and Thompson (2002) \cite{higgins2002quantifying} and  Higgins et al. (2003) \cite{higgins2003measuring} have been cited more than 32,000 and 57,000 times, respectively. Despite of its huge popularity, there were evidences in the literature reporting the limitations
of the $I^2$ statistic. In particular, R{\"u}cker et al. (2008) \cite{rucker2008undue} found that the $I^2$ statistic always increases rapidly to 1 when the sample sizes are large, regardless of whether or not the heterogeneity between the studies is clinically important. For other discussions on the $I^2$ statistic as a measure of heterogeneity, one may refer to, for example, Riley et al. (2016)\cite{riley2016external}, IntHout et al. (2016) \cite{inthout2016plea}, Borenstein et al. (2017) \cite{borenstein2017basics}, Sangnawakij et al. (2019) \cite{sangnawakij2019meta}, Holling et al. (2020)  \cite{holling2020evaluation}, and the references therein. This motivates us to further explore the characteristics of the $I^2$ statistic as a measure of heterogeneity for meta-analysis.

To answer this question, we first present a motivating example to demonstrate that the $I^2$ statistic was defined to quantify the heterogeneity between the observed effect sizes rather than that between the study populations. In view of this, we regard the $I^2$ statistic as a relative measure of heterogeneity. We further draw a connection between the one-way analysis of variance (ANOVA) and the random-effects meta-analysis, and subsequently introduce an alternative measure for quantifying the heterogeneity in the random-effects model, which is independent of study sample sizes and can serve as an absolute measure of heterogeneity. For details, see Section \ref{sec3.2} for the defined ${\rm ICC}_{\rm MA}$ in formula (\ref{icc}). To move forward, the statistical properties of ${\rm ICC}_{\rm MA}$ are also derived that explore the distinction between our new measure and the existing measures, together with an asymptotically unbiased estimator of the unknown ${\rm ICC}_{\rm MA}$ based on ANOVA. Lastly and most importantly, we also manage to provide an easy-to-implement estimator, namely the $I^2_{\rm A}$ statistic, to estimate ${\rm ICC}_{\rm MA}$ based on the $Q$ statistic, in a way similar for $I^2$ in (\ref{i2}) to estimate ${\rm ICC}_{\rm HT}$.

The remainder of the paper is organized as follows. In Section \ref{sec2}, we give a motivating example to illustrate that ${\rm ICC}_{\rm HT}$ heavily depends on the study sample sizes. In Section Section \ref{sec3}, by drawing a close connection between ANOVA and the random-effects meta-analysis, we introduce an alternative measure for quantifying the heterogeneity between the studies, namely ${\rm ICC}_{\rm MA}$, and then provide an ANOVA-based method to estimate this measure. In Sections \ref{sec4} to \ref{sec6}, we further derive the easy-to-implement $I^2_{\rm A}$ statistic to estimate the new heterogeneity measure ${\rm ICC}_{\rm MA}$, using the $Q$ statistic under three common scenarios with the raw mean, the mean difference, or the standardized mean difference as the effect size, respectively. While for practical implementation, real data analysis and numerical results are also presented for each scenario. Finally, we conclude the paper in Section \ref{sec7} and provide the technical details in the Appendix.

\section{A motivating example}\label{sec2}
In this section, we illustrate how ${\rm ICC}_{\rm HT}$ in (\ref{icch}) varies along with the sample sizes, and so may not be able to serve as a measure of heterogeneity between the study populations. To confirm this claim, we first consider a motivating example of three studies with data generated from normal populations $N(-0.05,1)$, $N(0,1)$ and $N(0.05,1)$, respectively. From the top-left panel of Figure \ref{fig1}, it is evident that the three study populations are largely overlapped. Taken the three study means as a random sample, the between-study variance can be estimated by the sample variance as $\tilde\tau^2=\{(-0.05-0)^2+(0-0)^2+(0.05-0)^2\}/2=0.0025$.
\begin{figure}[htp!]
	\begin{center}
		\begin{tabular}{ccc}
			\psfig{figure=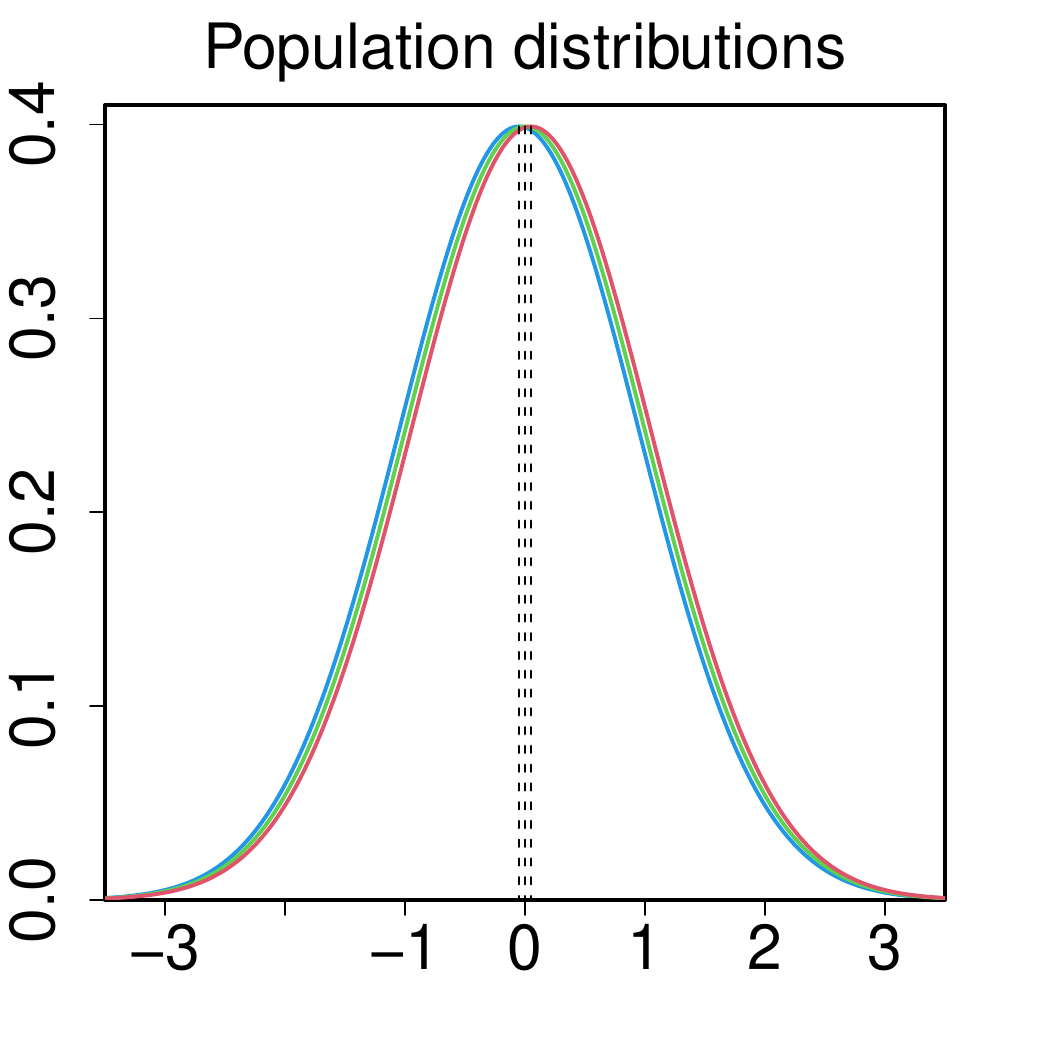,width=2.1in,angle=0}&
			\psfig{figure=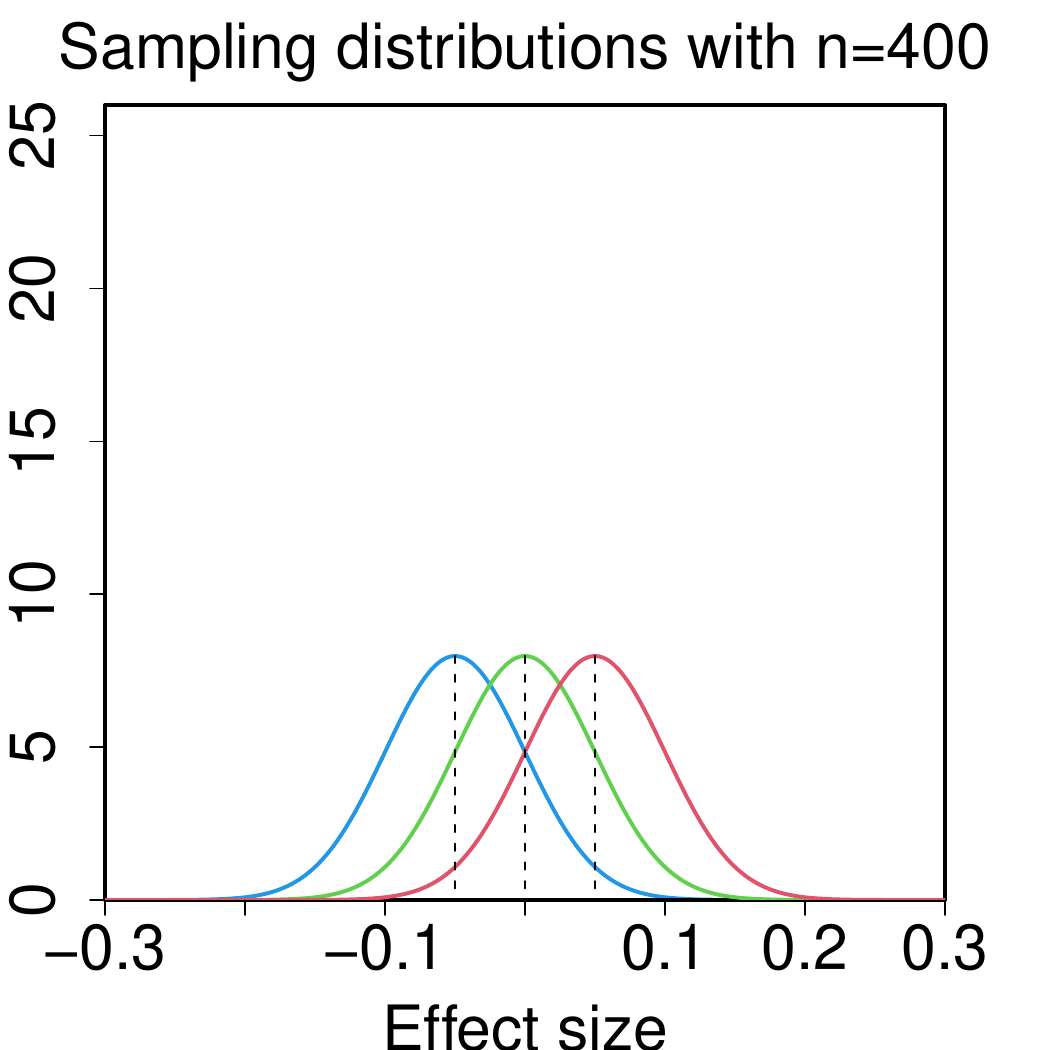,width=2.1in,angle=0}&
			\psfig{figure=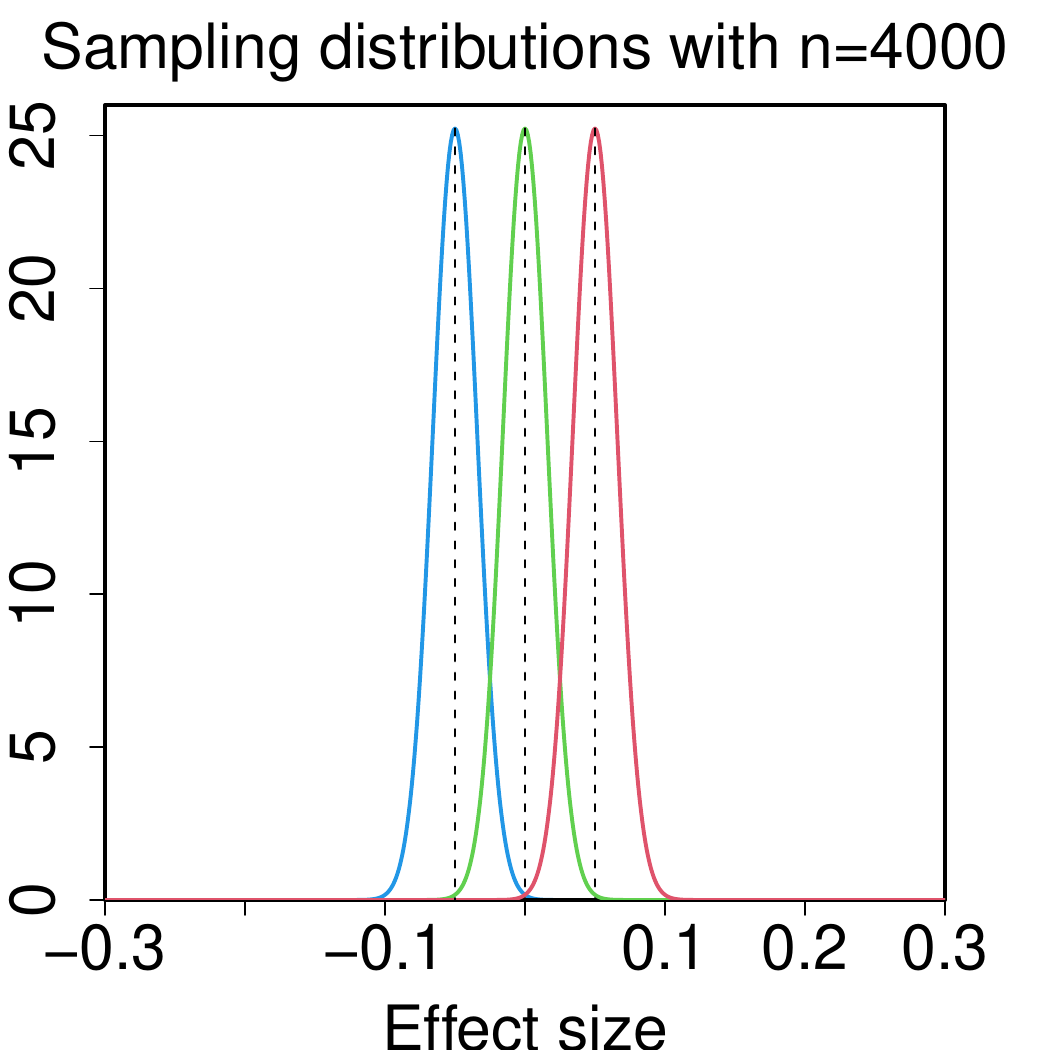,width=2.1in,angle=0}
		\end{tabular}{\caption{Population distributions of the three studies and the sampling distributions of the observed effect sizes. Left panel: Population distributions are $N(-0.05,1)$ in blue, $N(0,1)$ in green and $N(0.05,1)$ in red, respectively. Middle panel: Sampling distributions are $N(-0.05,0.0025)$, $N(0,0.0025)$ and $N(0.05,0.0025)$, respectively. Right panel: Sampling distributions  are $N(-0.05,0.00025)$, $N(0,0.00025)$ and $N(0.05,0.00025)$, respectively.}\label{fig1}}
	\end{center}
\end{figure}

To explain why ${\rm ICC}_{\rm HT}$ is not a measure of heterogeneity between the study populations, we consider two scenarios to meta-analyze the three studies, with the population means being treated as the effect sizes. The first scenario assumes $n=400$ patients in each study. By taking the sample means, the sampling distributions of the observed effect sizes are thus $N(-0.05,0.0025)$, $N(0,0.0025)$ and $N(0.05,0.0025)$, respectively, yielding $\sigma_y^2=0.0025$ as the common within-study variance. Further by the definition in (\ref{icch}), we have 
\beqrs
{\rm ICC}_{\rm HT}\approx\frac{0.0025}{0.0025+0.0025}=50\%.
\eeqrs
In the second scenario, we consider $n=4000$ for each study. This leads to the sampling distributions of the observed effect sizes as $N(-0.05,0.00025)$, $N(0,0.00025)$ and $N(0.05,0.00025)$, respectively. Further by $\sigma_y^2=0.00025$, the measure of heterogeneity is
\beqrs
{\rm ICC}_{\rm HT}\approx\frac{0.00025}{0.00025+0.0025}=90.9\%.
\eeqrs
Finally, for ease of comparison, we also plot the sampling distributions of the observed effect sizes in Figure \ref{fig1} for the two hypothetical scenarios with varying study sample sizes.

The above example clearly shows that ${\rm ICC}_{\rm HT}$, defined in (\ref{icch}) by  Higgins and Thompson (2002), measures the heterogeneity between the observed effect sizes and thus heavily depends on the study sample sizes. In other words, ${\rm ICC}_{\rm HT}$ is a relative measure of heterogeneity for meta-analysis. Consequently, as a sample estimate of ${\rm ICC}_{\rm HT}$, the $I^2$ statistic is also heavily dependent on the sample sizes. This coincides with the observations by R{\"u}cker et al. (2008) \cite{rucker2008undue}. Specifically, in our motivating example, ${\rm ICC}_{\rm HT}$ increases rapidly to about $90\%$ when the sample sizes are 4000, even though it is evident that the three populations are largely overlapped with each other. To summarize, when the study sample sizes $n_i$ are large enough, it will always yield an $I^2$ value being close to 1. On the other hand, compared with the population variance 1, the differences between the three study means $(-0.05, 0, 0.05)$ may not be clinically important. To support this claim, we note that the Scientific Committee of the European Food Safety Authority have also emphasized the importance of assessing the biological differences\cite{efsa2011statistical}. This hence motivates us to introduce an alternative measure that quantifies the heterogeneity between the study populations directly, in a way to reduce the influence of sample sizes.

\section{A new measure of heterogeneity and the $I^2_{\rm A}$ statistic}\label{sec3}
To further explore the characteristics of ${\rm ICC}_{\rm HT}$, we also draw in this section an interesting connection between one-way analysis of variance (ANOVA) and meta-analysis. And on the basis of that, a new measure for quantifying the heterogeneity between the study populations will be introduced, and moreover by studying its statistical properties, it is also clarified why it can add new value to meta-analysis. Lastly for completeness, we also provide an asymptotically unbiased estimator, namely the $I^2_{\rm ANOVA}$ statistic, to estimate the new measure of heterogeneity in Section \ref{sec3.3}. Nevertheless, as will be seen, the $I^2_{\rm ANOVA}$ statistic may not be easy to implement for practitioners, which motivates us to further propose a much simpler and more elegant estimator in  Sections \ref{sec4} to \ref{sec6} based on the $Q$ statistic. For readers who are not familiar with ANOVA, Section 3.3 can be skipped without affecting the subsequent reading.

\subsection{Connection between ANOVA and meta-analysis}\label{sec3.1}
To introduce the one-way ANOVA, we let $y_{ij}$ be the $j$th observation in the $i$th population, $i=1,\dots,k$ and $j=1,\dots,n_i$, where $k$ is the number of studies and $n_i$ are the study sample sizes from each population. The random-effects ANOVA for the observed data is then
\beqr\label{IPD}
y_{ij}=\mu+\delta_i+\xi_{ij},\quad \delta_i\stackrel{\text{i.i.d.}}{\sim} N(0,\tau^2), \quad \xi_{ij}\stackrel{\text{i.i.d.}}{\sim} N(0,\sigma^2),
\eeqr
where $\mu$ is the grand mean, $\delta_i$ are the treatment effects, and $\xi_{ij}$ are the random errors. We further assume that $\delta_i$ are i.i.d. normal random variables with mean 0 and variance $\tau^2\geq 0$, $\xi_{ij}$ are i.i.d. normal random errors with mean 0 and variance $\sigma^2>0$, and that $\delta_i$ and $\xi_{ij}$ are independent of each other. 
In addition, we refer to $\mu_i=\mu+\delta_i$ as the individual means, $\tau^2$ as the between-study variance, $\sigma^2$ as the common error variance for all $k$ populations, and $\tau^2+\sigma^2$ as the total variance of each observation.

To draw a close connection between ANOVA and meta-analysis, we consider a hypothetical scenario in which the experimenter first computed the sample mean and its variance for each population, namely $y_i=\sum_{j=1}^{n_i} y_{ij}/n_i$ and $\hat\sigma_{y_i}^2=\sum_{j=1}^{n_i}(y_{ij}-y_i)^2/\{n_i(n_i-1)\}$ for $i=1,\dots,k$, and then reported these summary data rather than the raw data to the public. In practice, there are reasons why one must do so, including, for example, due to the privacy protection for which the individual patient data cannot be released. Under such a scenario, if some researchers want to re-analyze the experiment using only the publicly available data, it then yields a new random-effects model as
\beqr\label{REM}
y_i=\mu+\delta_i+\epsilon_i,\quad \delta_i\stackrel{\text{i.i.d.}}{\sim} N(0,\tau^2), \quad \epsilon_i\stackrel{\text{ind}}{\sim} N(0,\sigma^2/n_i),
\eeqr
where $y_i$ are the sample means, $\mu$ and $\delta_i$ are the same as defined in model (\ref{IPD}), and $\epsilon_i=\sum_{j=1}^{n_i}\xi_{ij}/n_i$ are independent random errors with mean 0 and variance $\sigma^2/n_i$, where $i=1,\dots,k$. Now from the point of view of meta-analysis, if we treat $y_i$ as the reported effect sizes and $\hat{\sigma}^2_{y_i}$ as the within-study variances representing $\sigma^2/n_i$, then model (\ref{REM}) is essentially the same as the random-effects model in (\ref{meta}). This interesting connection shows that, when the ANOVA model with raw data only releases the summary data to the public, it will then yield a meta-analysis model with summary data.

For ease of comparison, we also summarize some key components in Table \ref{table1} for both the ANOVA model in (\ref{IPD}) and the meta-analysis model in (\ref{REM}). For the meta-analysis model, under the assumption that the within-study variances, i.e. $\sigma^2/n_i$, are all equal, Higgins and Thompson (2002) \cite{higgins2002quantifying} interpreted the measure of heterogeneity as the proportion of total variance that is ``between the studies". More specifically, by the last column of Table \ref{table1}, they introduced the measure of heterogeneity for meta-analysis as in (\ref{icch}),
where $\sigma_{y}^2=\sigma^2/n_i$ is the common within-study variance for the observed effect sizes. This clearly explains why ${\rm ICC}_{\rm HT}$ will be heavily dependent on the study sample sizes. When the sample sizes go to infinity, the within-study variances will converge to zero so that ${\rm ICC}_{\rm HT}$ will increase to 1, as having been observed in R{\"u}cker et al. (2008) \cite{rucker2008undue}. This also coincides with our motivating example in Section \ref{sec2} that, when the sample size varies from 400 to 4000, their measure of heterogeneity will increase from $50\%$ to about $90\%$, regardless of whether or not the heterogeneity between the studies is clinically important.

For the ANOVA model, it is well known that the intraclass correlation coefficient (ICC) is the most commonly used measure of heterogeneity \citep{fisher1925,smith1957estimation,donner1979use,mcgraw1996forming}, which interprets the proportion of total variance that is ``between populations". More specifically, by Table \ref{table1}, ${\rm ICC}$ can be expressed as
\beqr\label{rho}
{\rm ICC}=\frac{\tau^2}{{\rm var}(y_{ij})}=\frac{\tau^2}{\tau^2+\sigma^2}.
\eeqr
As shown in the hypothetical scenario, the ANOVA model in (\ref{IPD}) and the meta-analysis model in (\ref{REM}) are, in fact, modeling the same populations, even though one uses the raw data and the other uses the summary data. In the special case when the mean value is taken as the effect size, it is known that the sample mean is a sufficient and complete statistic for the normal mean; in other words, the raw data and the summary data contain exactly the same information regarding the effect size. With this insight, we expect that the measures of heterogeneity between the study populations for the two models should also be the same, regardless of whether the raw data or the summary data are being used.

\subsection{An intrinsic measure of heterogeneity}\label{sec3.2}
Inspired by the intrinsic connection between ANOVA and meta-analysis, we now follow the same assumption as in ANOVA that the population variances $n_i\sigma_{y_i}^2$ are all equal. For ease of presentation, we also denote the common study population variance as $\sigma_{\rm pop}^2$. Then by following ICC in (\ref{rho}) for ANOVA, we propose the following measure of heterogeneity for meta-analysis:
\beqr\label{icc}
{\rm ICC}_{\rm MA}=\frac{\tau^2}{{\rm var}(y_{ij})}=\frac{\tau^2}{\tau^2+\sigma_{\rm pop}^2}.
\eeqr
Note that the range of ${\rm ICC}_{\rm MA}$ is always within the interval $[0,1)$. Regarding the rationale of ${\rm ICC}_{\rm MA}$ for meta-analysis, one may also refer to the proposed measure in Sangnawakij et al.
(2019) \cite{sangnawakij2019meta}. And as mentioned in Section 2, a common population variance can be a more reasonable assumption for meta-analysis compared to a common within-study variance for all studies, in a way to mitigate the impact caused by the study sample sizes.

To further study the properties of ${\rm ICC}_{\rm MA}$ and explain why it can serve as an absolute measure of heterogeneity for meta-analysis, we first present the three statistical properties of ${\rm ICC}_{\rm HT}$ as follows.
\begin{enumerate}
	\item[(i)] \textit{Monotonicity.} ${\rm ICC}_{\rm HT}$ is a monotonically increasing function of the ratio $\tau^2/\sigma_y^2$. When the common within-study variance $\sigma_y^2$ is fixed, ${\rm ICC}_{\rm HT}$ will solely increase with the between-study variance $\tau^2$. This property was referred to as the ``dependence on the extent of heterogeneity" by Higgins and Thompson (2002) \cite{higgins2002quantifying}.
	
	\item[(ii)] \textit{Location and scale invariance.}
	${\rm ICC}_{\rm HT}$ is not affected by the location and scale of the effect sizes. This property was referred to as the ``scale invariance" by Higgins and Thompson (2002) \cite{higgins2002quantifying}.
	
	\item[(iii)] \textit{Study size invariance.} ${\rm ICC}_{\rm HT}$ is not affected by the total number of studies $k$. This property was referred to as the ``size invariance" by Higgins and Thompson (2002) \cite{higgins2002quantifying}.
\end{enumerate}
Thanks to the above properties, the $I^2$ statistic is nowadays the most popular measure for quantifying the heterogeneity in meta-analysis, compared to other existing measures including $Q$ and $\tau^2$. Nevertheless, we do wish to point out that the ``size invariance" in their property (iii) only represents the study size invariance but not includes the sample size invariance. As shown in the motivating example and also from the historical evidence in the literature, ${\rm ICC}_{\rm HT}$ does suffer from a heavy dependence on the study sample sizes.

While for the new measure of heterogeneity in (\ref{icc}), we show in  \hyperref[appA]{Appendix A} that ${\rm ICC}_{\rm MA}$ shares the following four properties:
\begin{enumerate}
	\item[(i$'$)] \textit{Monotonicity.} ${\rm ICC}_{\rm MA}$ is a monotonically increasing function of the ratio $\tau^2/\sigma_{\rm pop}^2$. When the common population variance $\sigma_{\rm pop}^2$ is fixed, ${\rm ICC}_{\rm MA}$ will solely increase with the between-study variance $\tau^2$.
	
	\item[(ii$'$)] \textit{Location and scale invariance.}
	${\rm ICC}_{\rm MA}$ is not affected by the location and scale of the effect sizes.
	
	\item[(iii$'$)] \textit{Study size invariance.} ${\rm ICC}_{\rm MA}$ is not affected by the total number of studies $k$.
	
	\item[(iv$'$)] \textit{Sample size invariance.} ${\rm ICC}_{\rm MA}$ is not affected by the sample size $n_i$ of each study.
\end{enumerate}
Note that the first three properties for ${\rm ICC}_{\rm MA}$ are essentially the same as those for ${\rm ICC}_{\rm HT}$. While for the importance of property (iv$'$), let us illustrate again using the motivating example in Section \ref{sec2}. Under the assumption of a common population variance, the term $\sigma_{\rm pop}^2$ remains constant at 1 no matter how the sample sizes vary. Further by (\ref{icc}), the value of ${\rm ICC}_{\rm MA}$ under each scenario will always be $0.0025/(0.0025+1)\approx 0.25\%$, indicating that the three study populations are indeed highly overlapped with a small amount of heterogeneity. To conclude, it is because of the sample size invariance in property (iv$'$) that distinguishes our new ${\rm ICC}_{\rm MA}$ from the existing ${\rm ICC}_{\rm HT}$, which also perfectly explains why ${\rm ICC}_{\rm MA}$ can serve as a new measure for quantifying the heterogeneity between the study populations. Due to its sample size invariance, we regard ${\rm ICC}_{\rm MA}$ as an absolute measure of heterogeneity.

\subsection{Estimation of ${\rm ICC}_{\rm MA}$ based on ANOVA}\label{sec3.3}
	
In ANOVA, there has been extensive and well-established research on the estimation of ICC. For easy reference, we have also provided a brief review in \hyperref[appB]{Appendix B}. Among the existing methods, it is known that the ANOVA estimator is the most widely used thanks to its straightforwardness and effectiveness. Inspired by this, we also propose an ANOVA-based estimator for ${\rm ICC}_{\rm MA}$ in the framework of meta-analysis.

Following the random-effects ANOVA in (\ref{IPD}), the total variance of the observations is given by $\sum_{i=1}^k\sum_{j=1}^{n_i}(y_{ij}-\bar y)^2$, which can be divided into two components as the sum of squares between the populations and the error sum of squares within the populations. Based on this variance partitioning,  Cochran (1939) \cite{cochran1939use} derived the method of moments estimators of $\tau^2$ and $\sigma^2$, and then by plugging them into formula (\ref{rho}), it yields the well known ANOVA estimator for the unknown ICC. In parallel, following the random-effects model for meta-analysis in (\ref{meta}), we first assume that $\hat\sigma_{y_i}^2$ are the estimated within-study variance from each study, as also mentioned in Section \ref{sec3.1}. We further define the mean square between the populations (MSB) as 
\beqr\label{msb}
{\rm MSB}_{\rm MA}=\frac{1}{k-1}\sum_{i=1}^k\left\{n_i\left(y_i-\bar y\right)^2\right\},
\eeqr
where $\bar y=\sum_{i=1}^k(n_iy_i)/\sum_{i=1}^kn_i$, and the mean square within the populations (MSW) as
\beqr\label{msw}
{\rm MSW}_{\rm MA}=\frac{1}{\sum_{i=1}^k\left(n_i-1\right)}\sum_{i=1}^k\left\{n_i\left(n_i-1\right)\hat\sigma_{y_i}^2)\right\}.
\eeqr
Moreover, let 
\beqr\label{nbar}
\tilde n=\frac{1}{k-1}\left(\sum_{i=1}^kn_i-\sum_{i=1}^kn^2_i/\sum_{i=1}^kn_i\right)
\eeqr
be the adjusted mean sample size \cite{thomas1978interval} that accounts for the variation of the sample sizes from different studies. Then by the same method for estimating ICC, our new estimator for ${\rm ICC}_{\rm MA}$ is given as
\beqr\label{iqb}
I^2_{\rm ANOVA}=\max\left\{\frac{{\rm MSB}_{\rm MA}-{\rm MSW}_{\rm MA}}{{\rm MSB}_{\rm MA}+(\tilde n-1){\rm MSW}_{\rm MA}},0\right\}.
\eeqr
Similar to the $I^2$ statistic in (\ref{i2}), the maximum operation is taken to avoid a negative estimate. For more details on the derivation of $I^2_{\rm ANOVA}$, one may refer to \hyperref[appC]{Appendix C}.

Up to now, we have used the generic notation $y_i$ as the observed effect size, together with its standard error $\hat\sigma_{y_i}$ and the sample size $n_i$. Note that this is the simplest scenario, in which the effect size is represented by the mean $y_i$ from each study with only one arm. In addition to the mean, two other commonly used effect sizes for continuous outcomes are the mean difference (MD) and the standardized mean difference (SMD), which are applicable to meta-analysis of studies with two arms. In the next two paragraphs, we show that the $I^2_{\rm ANOVA}$ statistic in (\ref{iqb}) can be directly generalized to handle these two scenarios.

For a meta-analysis of MD, the summary statistics for the $i$th study often consist of the observed MD $y_i$, the sample sizes $n_i^T$ and $n_i^C$, and the standard errors $\hat\sigma_{y_i^T}$ and $\hat\sigma_{y_i^C}$ associated with the treatment and control groups. With these notations, the mean square within the populations can be computed as
\beqrs
{\rm MSW}_{\rm MA}&=&\frac{\sum_{i=1}^k\left\{n_i^T\left(n_i^T-1\right)\hat\sigma^2_{y^T_i}+n_i^C\left(n_i^C-1\right)\hat\sigma^2_{y^C_i}\right\}}{\sum_{i=1}^k\left(n_i^T+n_i^C\right)-2k}.
\eeqrs
In addition, by defining the effective sample size as $n_i=1/(1/n_i^T+1/n_i^C)$ for each study, the formulas (\ref{msb}) and (\ref{nbar}) can also be directly followed to calculate ${\rm MSB}_{\rm MA}$ and the adjusted mean sample size $\tilde n$. Lastly by (\ref{iqb}), we can derive the $I^2_{\rm ANOVA}$ statistic as the estimated measure of heterogeneity.
For more details about the meta-analysis of MD including the statistical models and the underlying assumptions, one may refer to \hyperref[appD]{Appendix D}.

For a meta-analysis of SMD, the summary statistics will instead report the observed SMD $y_i$ for each study, together with the sample sizes $n_i^T$ and $n_i^C$, and the standard errors $\hat\sigma_{y_i^T}$ and $\hat\sigma_{y_i^C}$. Then to compute the $I^2_{\rm ANOVA}$ statistic by (\ref{iqb}), we note that the same procedure as that for MD can still be followed to determine the values of ${\rm MSB}_{\rm MA}$ and $\tilde n$. And moreover, we can also set ${\rm MSW}_{\rm MA}$ directly to 1 since the observed effect sizes are already standardized. For a comprehensive understanding of the model specifications for meta-analysis of SMD, one may refer to \hyperref[appE]{Appendix E}.

\section{The $I^2_{\rm A}$ statistic for the mean}\label{sec4}

Recall that to estimate ${\rm ICC}_{\rm HT}$, Higgins and Thompson (2002) \cite{higgins2002quantifying} proposed the easy-to-implement $I^2$ statistic based on the $Q$ statistic. Nevertheless, for our new measure of heterogeneity ${\rm ICC}_{\rm MA}$, the ANOVA-based estimator in (\ref{iqb}) is somewhat complicated and may not be familiar for meta-analysts. This motivates us to further propose a new estimator of ${\rm ICC}_{\rm MA}$, referred to as $I^2_{\rm A}$, which turns out to have a similar form as the $I^2$ statistic. More specifically, we will present the $I^2_{\rm A}$ statistic in Sections \ref{sec4} to \ref{sec6} for the three effect sizes including the mean, MD and SMD, respectively, followed by real data analyses and simulation studies that compare the numerical performance of the $I^2$, $I^2_{\rm ANOVA}$ and $I^2_{\rm A}$ statistics.
	
By (\ref{icch}), ${\rm ICC}_{\rm HT}$ is defined based on the assumption of a common within-study variance. When this assumption does not hold, as pointed out in the literature, the common within-study variance $\sigma_y^2$ can be replaced by $\tilde\sigma_y^2$ in (\ref{i2}) as an average of the $k$ within-study variances. Now for ${\rm ICC}_{\rm MA}$ in (\ref{icc}), we have $\sigma^2_{\rm pop}=n_i\sigma^2_{y_i}$, and consequently, $w_i=1/\sigma_{y_i}^2=n_i/\sigma^2_{\rm pop}$. Then by letting $\tilde n=(\sum_{i=1}^kn_i-\sum_{i=1}^kn_i^2/\sum_{i=1}^kn_i)/(k-1)$ be the adjusted mean sample size as in (\ref{nbar}), an identity between $\tilde\sigma_y^2$ and $\sigma^2_{\rm pop}$ can be established as follows:
\beqr\label{identity}
\tilde\sigma_y^2=\frac{k-1}{\sum_{i=1}^kw_i-\sum_{i=1}^kw_i^2/\sum_{i=1}^kw_i}=\frac{1}{\tilde n}\sigma^2_{\rm pop}.
\eeqr
In addition, since $E\{Q/(k-1)-1\}=\tau^2/\tilde\sigma_y^2$ by Higgins and Thompson (2002),\cite{higgins2002quantifying} it follows that 
\beqrs
E\left(\frac{Q-(k-1)}{(k-1)\tilde n}\right)=\frac{\tau^2}{\sigma_{\rm pop}^2},
\eeqrs
which leads to a method of moments estimator of $\tau^2/\sigma_{\rm pop}^2$ as $\{Q-(k-1)\}/\{(k-1)\tilde n\}$. Lastly, by noting that  ${\rm ICC}_{\rm MA}=(\tau^2/\sigma_{\rm pop}^2)/(\tau^2/\sigma_{\rm pop}^2+1)$, our plug-in estimator for ${\rm ICC}_{\rm MA}$ is given as
\beqr\label{iq1}
I_{\rm A}^2=\max\left\{\frac{Q-(k-1)}{Q+(k-1)(\tilde n-1)},0\right\},
\eeqr
where, as usual, the maximum operation is kept to avoid a negative estimate. By comparing (\ref{i2}) and (\ref{iq1}), we note that the difference between the $I^2$ and $I_{\rm A}^2$ statistics is purely on the term $(k-1)(\tilde n-1)$, which is a function of the study sample sizes and the number of studies. In the special case when $\tilde n=1$, the two statistics will be exactly the same.

For more insights on how the estimated heterogeneity is adjusted from the relative measure to the absolute measure by the study sample sizes through $(k-1)(\tilde n-1)$, we summarize below a few interesting findings with the proofs in \hyperref[appa]{Appendix F}. 
\begin{enumerate}
	\item[(a)] Firstly, we have $(k-1)(\tilde n-1)\ge 0$, where the equality holds only when $n_i=1$ for all $k$ studies. Consequently, it yields that $$I_{\rm A}^2\leq I^2$$ under all the settings of meta-analysis with at least 2 studies.
	\item[(b)] For the balanced design where all sample sizes are equal to $n$, we have $\tilde n=n$ and moreover $(k-1)(\tilde n-1)=(k-1)(n-1)$. 
	When $k\to \infty$ and $n\to\infty$, we can further show that $Q/\{(k-1)(n-1)\} \to \tau^2/\sigma_{\rm pop}^2=O(1)$, indicating that the two terms in the denominator of (\ref{iq1}) are of the same asymptotic order. Moreover by Slutsky's theorem, $$I_{\rm A}^2 \to {\tau^2/\sigma_{\rm pop}^2 \over \tau^2/\sigma_{\rm pop}^2 +1}= {\rm ICC}_{\rm MA}<1. $$On the other hand, noting that $Q/(k-1)=O(n)$ for any fixed $k\geq 2$, we have $I^2 = 1-O(1/n) \to 1$ as $n\to\infty$. Taken together, it clearly explains why $I^2$ may asymptotically increase to 1, whereas our new $I_{\rm A}^2$ will not.
	\item[(c)] For the unbalanced design, it can also be shown that $(k-1)(\tilde n-1)$ is an increasing function of $n_i$ given that all other sample sizes are fixed, and moreover, $Q/\{(k-1)(\tilde n-1)\}\to \tau^2/\sigma_{\rm pop}^2$ when $k\to\infty$ and all $n_i\to \infty$. Consequently, we still have $I_{\rm A}^2 \to {\rm ICC}_{\rm MA}<1$ as that for the balanced case. In contrast, without the adjustment term $(k-1)(\tilde n-1)$ that is the same order of $Q$, the $I^2$ statistic will again, as observed in the literature, rapidly increase to 1 as the sample sizes are large.
	\item[(d)] Lastly, we can also express the $I_{\rm A}^2$ statistic in (\ref{iq1}) as
	\beqr\label{unequalv}
	I_{\rm A}^2=\frac{\hat\tau^2}{\hat\tau^2+\tilde n\tilde\sigma_y^2},
	\eeqr
	where $\hat\tau^2$ is the DerSimonian-Laird estimator as mentioned in Section \ref{sec1}. When the assumption of a common population variance holds, $\tilde n\tilde\sigma_y^2$ is exactly $\sigma_{\rm pop}^2$ by (\ref{identity}). Otherwise, we can still apply $\tilde n\tilde\sigma_y^2$ to compute the $I_{\rm A}^2$ statistic by (\ref{unequalv}). And similarly as $\tilde \sigma_y^2$ explained in B{\"o}hning et al. (2017),\cite{bohning2017some} $\tilde n\tilde\sigma_y^2$ is asymptotically equivalent to the adjusted mean sample size multiplied by the harmonic mean of the within-study variances.
\end{enumerate}

\subsection{Real data analysis}\label{sec4.2}
This section applies a real data example to illustrate the $I^2_{\rm A}$ statistic, and compare it with the $I^2_{\rm ANOVA}$ and $I^2$ statistics, for quantifying the heterogeneity between the studies. Specifically, we revisit a previous meta-analysis conducted by Jeong et al. (2014) \cite{jeong2014efficacy}, which investigated the stem cell-based therapy as a novel approach for the stroke treatment. Among various measures of efficacy and safety, we focus on the point difference in the National Institutes of Health Stroke Scale as the outcome. The summary data for a total of $k=10$ studies are presented in Table \ref{data}.

Treating $\hat\sigma^2_{y_i}$ in Table \ref{data} as the true values of $\sigma^2_{y_i}$, we have $\sum_{i=1}^{10}w_i=7.68$ and $\sum_{i=1}^{10} w_iy_i=-43.39$. This leads to Cochran's $Q$ statistic in (\ref{qs}) as $Q=106.26$. Moreover, by formula (\ref{i2}),
	\beqrs
	I^2=\max\left\{\frac{106.26-(10-1)}{106.26},0\right\}=0.92.
	\eeqrs
	In addition, the adjusted mean sample size can be computed as $\tilde n=8.97$. Then by formula (\ref{iq1}), it yields that
		\beqrs
	I_{\rm A}^2=\max\left\{\frac{106.26-(10-1)}{106.26+(10-1)(8.97-1)},0\right\}=0.55.
	\eeqrs 
	Lastly, to also include the $I_{\rm ANOVA}^2$ statistic for additional comparison, we have $\bar y = \sum_{i=1}^{10} n_iy_i / \sum_{i=1}^{10} n_i = -7.55$, ${\rm MSB}_{\rm MA}=189.83$, and ${\rm MSW}_{\rm MA}=25.81$.
	Consequently, by formula (\ref{iqb}),
	\beqrs
	I^2_{\rm ANOVA}=\max\left\{\frac{189.83-25.81}{189.83+(8.97-1)\times 25.81},0\right\}=0.41.
	\eeqrs
	To conclude, unlike the $I^2$ statistic that is very close to 1, the values of $I^2_{\rm A}$ and $I^2_{\rm ANOVA}$ are close to each other and they both indicate a moderate heterogeneity for the ten studies.
\begin{figure}[htp!]
	\begin{center}
		\begin{tabular}{cc}
			\psfig{figure=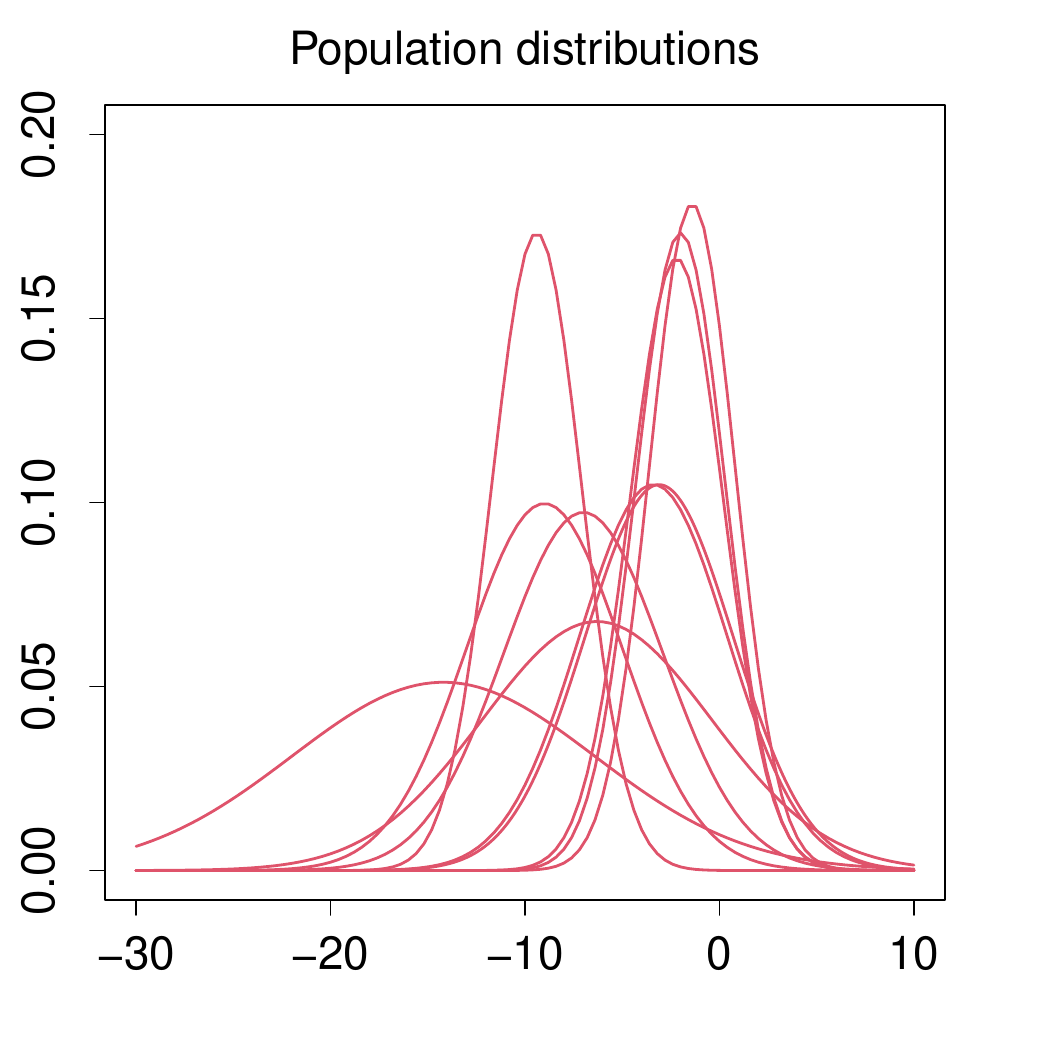,width=2.8in,angle=0}&
			\psfig{figure=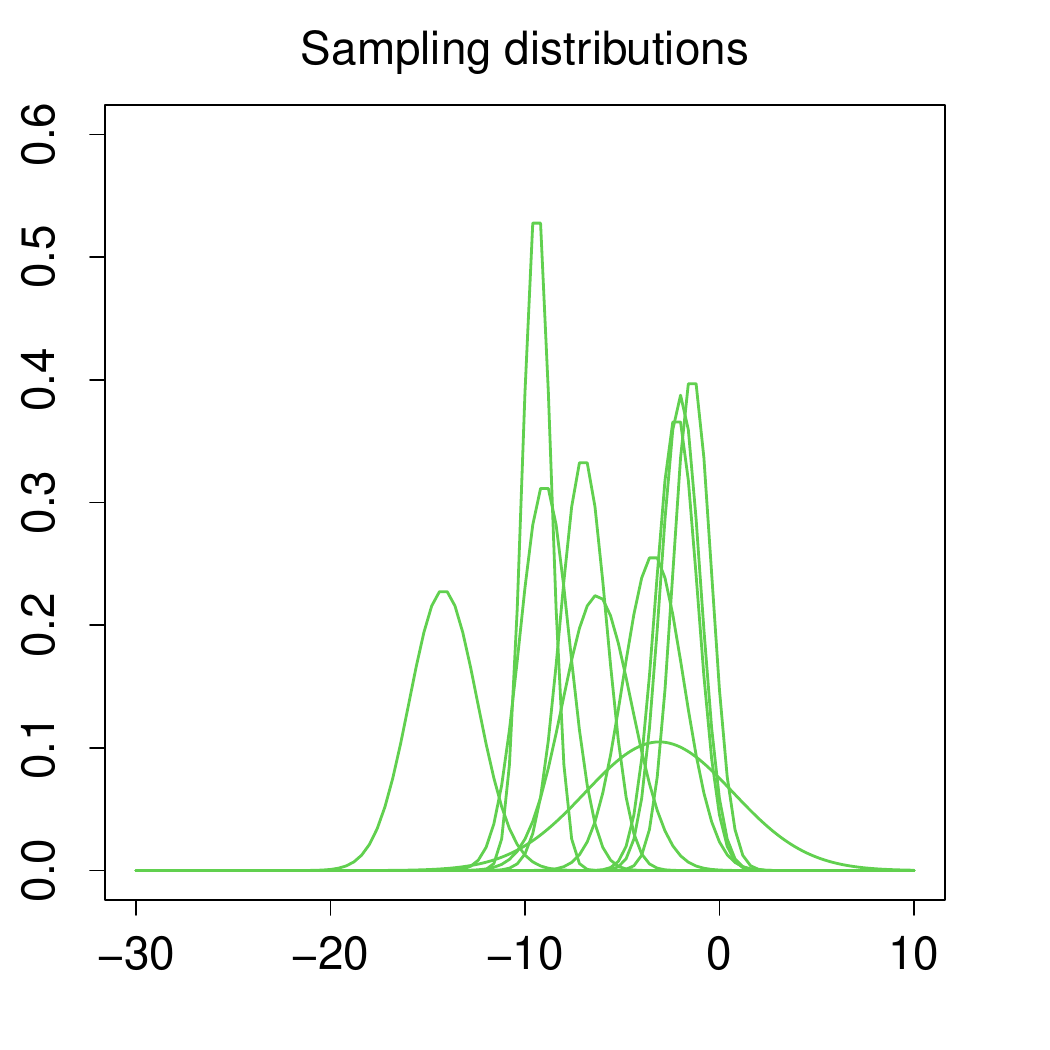,width=2.8in,angle=0}
		\end{tabular}{\caption{Population distributions of the ten studies and the sampling distributions of the observed effect sizes from Jeong et al. (2014). For each study, the population distribution is assumed to be normal with mean $y_i$ and variance $n_i\hat\sigma_{y_i}^2$. The sampling distribution of the effect size is assumed to be normal with mean $y_i$ and variance $\hat\sigma_{y_i}^2$.}\label{figapp}}
	\end{center}
\end{figure}

To further compare the $I^2_{\rm A}$ statistic and the $I^2$ statistic, as a common practice we assume that the ten studies are all normally distributed. Then by the reported means and variances, we plot their respective population distributions and the sampling distributions of the observed effect sizes in Figure \ref{figapp} for visualization. From the figure, it is evident that the ten studies are not very heterogeneous since most of the study populations are largely overlapped in the range roughly from -15 to 5, corresponding to a measure of 0.55 for the $I^2_{\rm A}$ statistic. By contrast, the sampling distributions of the observed effect sizes are less overlapped with each other, indicating a much higher heterogeneity at 0.92 by the $I^2$ statistic.

\subsection{Numerical results}\label{sec4.3}
To compare the numerical performance of the three statistics, we now conduct simulations based on the random-effects model (\ref{REM}) with $\mu=0$ and $\sigma^2=100$. For the between-study variance, we consider $\tau^2=9$ or 90 that corresponds to ${\rm ICC}_{\rm MA}$ as $9/(9+100)=0.083$ or $90/(90+100)=0.474$, respectively. Let also $k=3$ or 10 to represent the small or large number of studies included in the meta-analysis. For the sample size of each study, we consider the unbalanced design with the sample size of the $i$th study being $i*n$, where $i=1,\dots,k$ and the common $n$ ranges from 10 to 90. With each of the above settings, we then generate the raw data from model (\ref{REM}) and report the summary data $y_i$ and $\hat\sigma_{y_i}^2$ for the $k$ studies. Finally with $M=10,000$ repetitions, we present the boxplots of the $I^2_{\rm A}$, $I^2_{\rm ANOVA}$ and $I^2$ statistics, together with their mean values, in Figure \ref{c1}. From the figure, it is evident that the $I^2$ statistic has an increasing trend with the sample size $n$. This is consistent with what was observed in R{\"u}cker et al. (2008) \cite{rucker2008undue} that the $I^2$ statistic always increases rapidly to 1 when the sample sizes are large. By contrast, with each solid line representing the heterogeneity ${\rm ICC}_{\rm MA}$ between the study populations, we note that  $I^2_{\rm A}$ and $I^2_{\rm ANOVA}$ are not influenced by the sample size and also provide comparable estimates for ${\rm ICC}_{\rm MA}$ in terms of both bias and variance. And more interestingly, they are able to perform even better when the number of studies $k$ is larger.
\begin{figure}[htp!]
	\begin{center}
		\begin{tabular}{cc}
			\psfig{figure=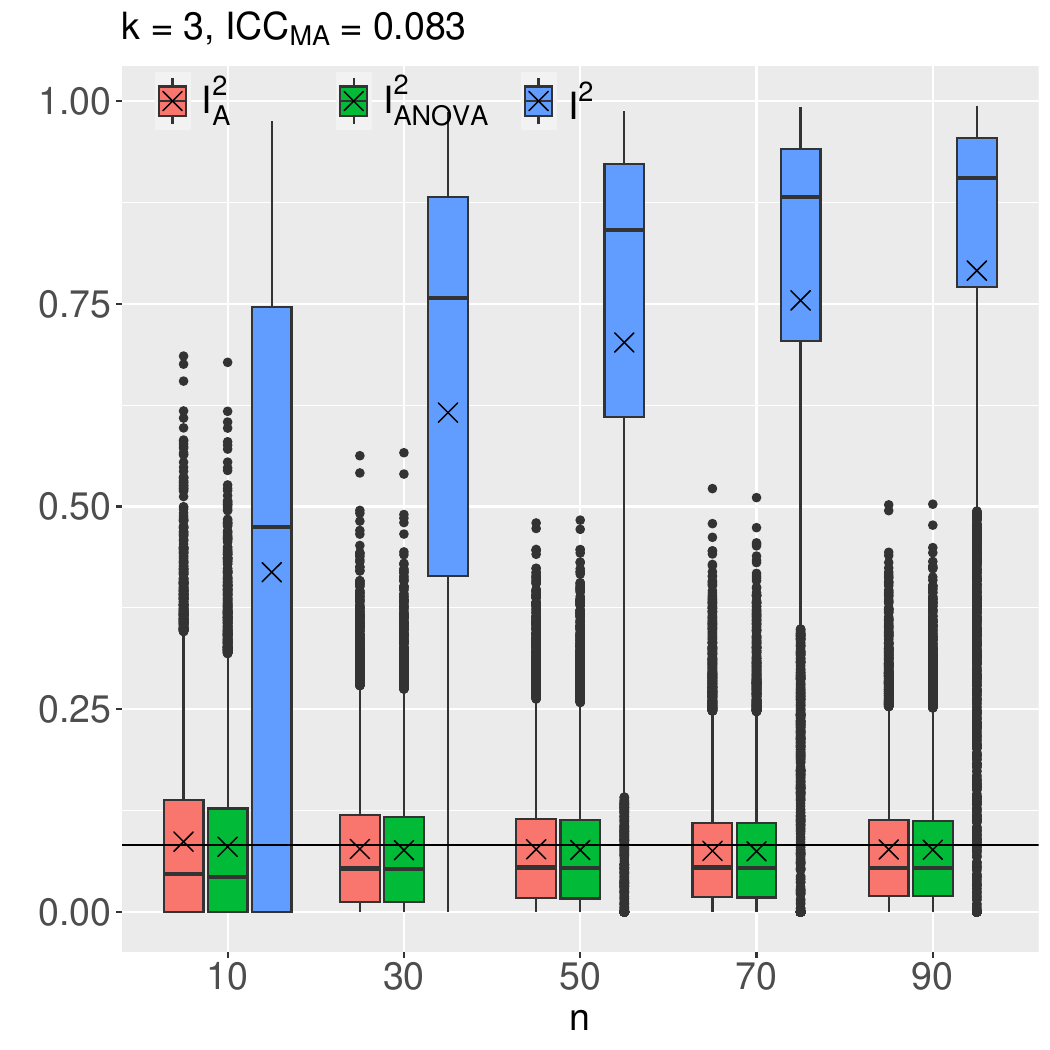,width=3.3in,angle=0}&
			\psfig{figure=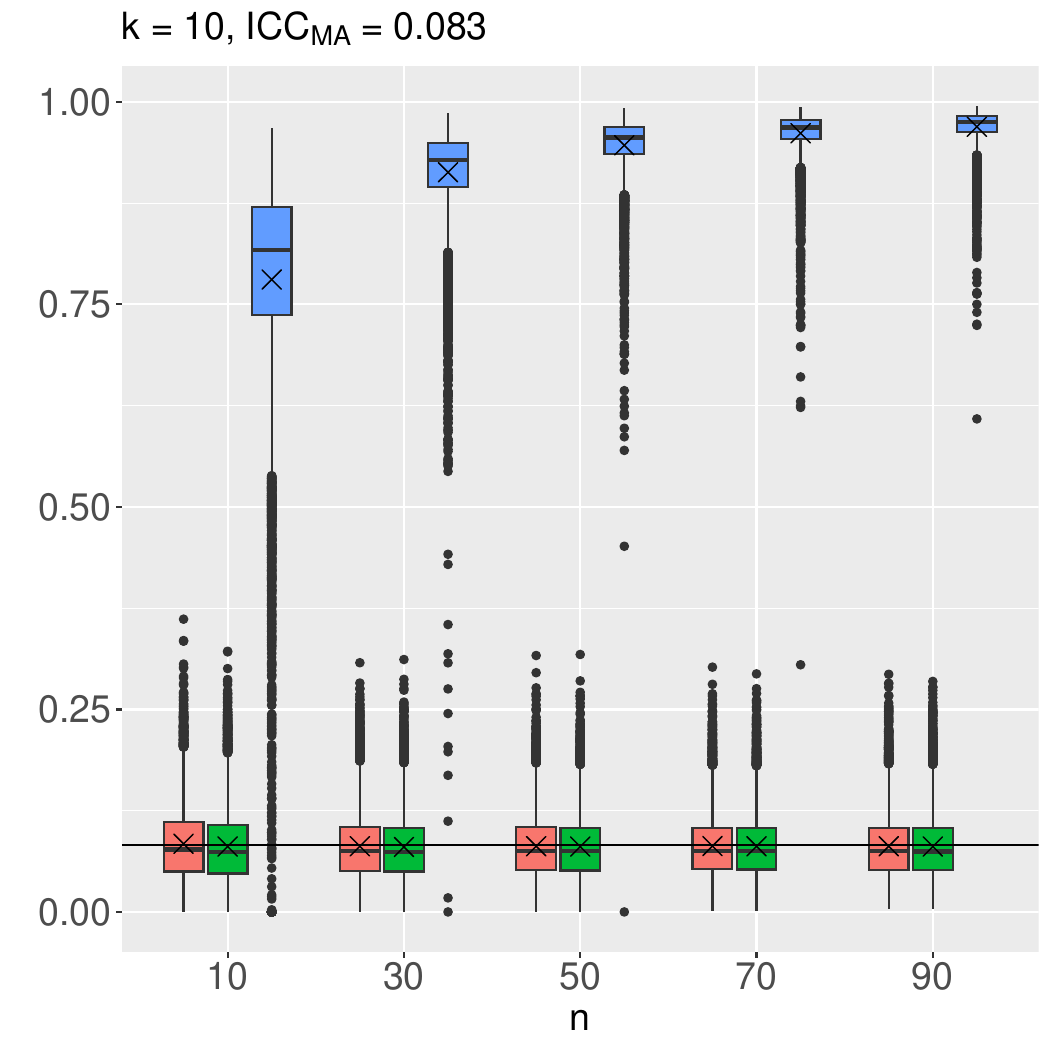,width=3.3in,angle=0}\\
			\psfig{figure=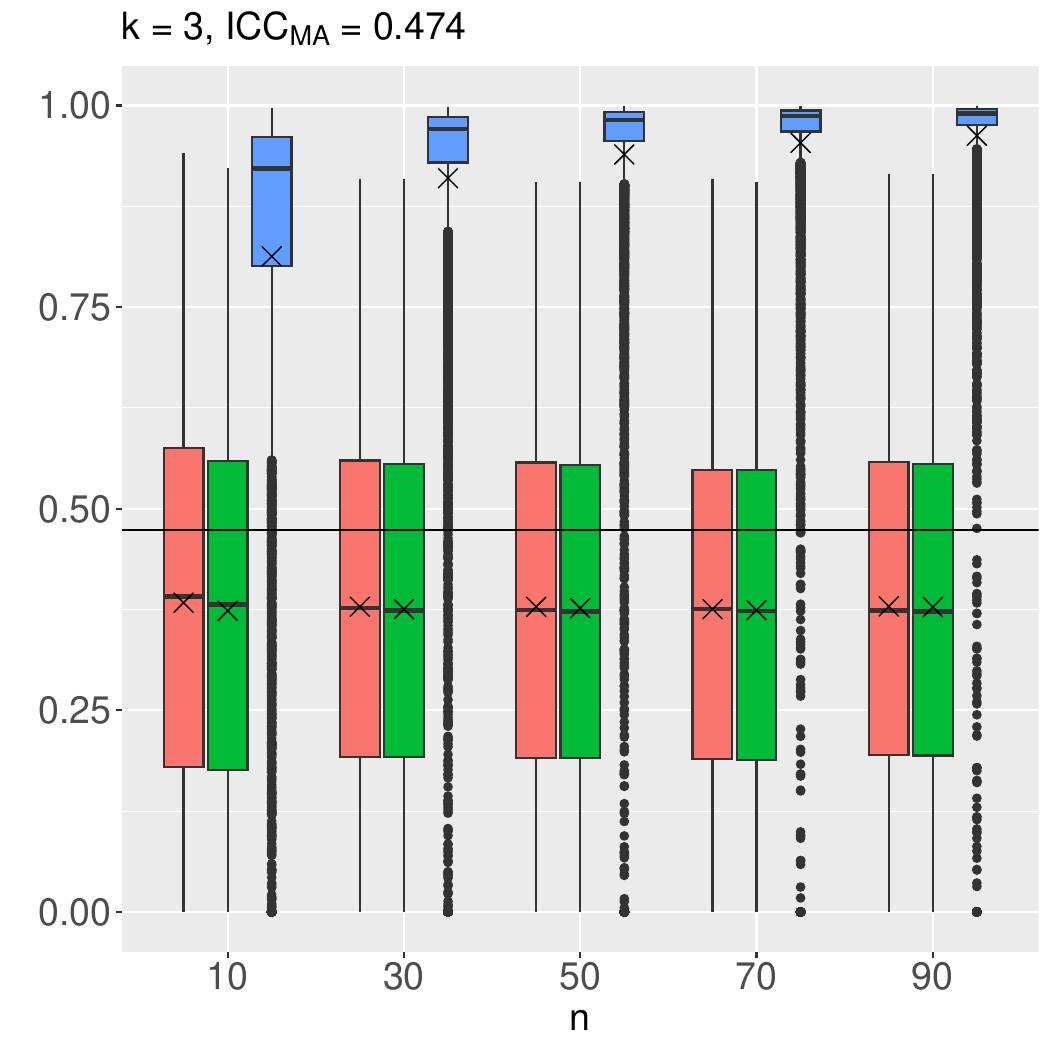,width=3.3in,angle=0}&
			\psfig{figure=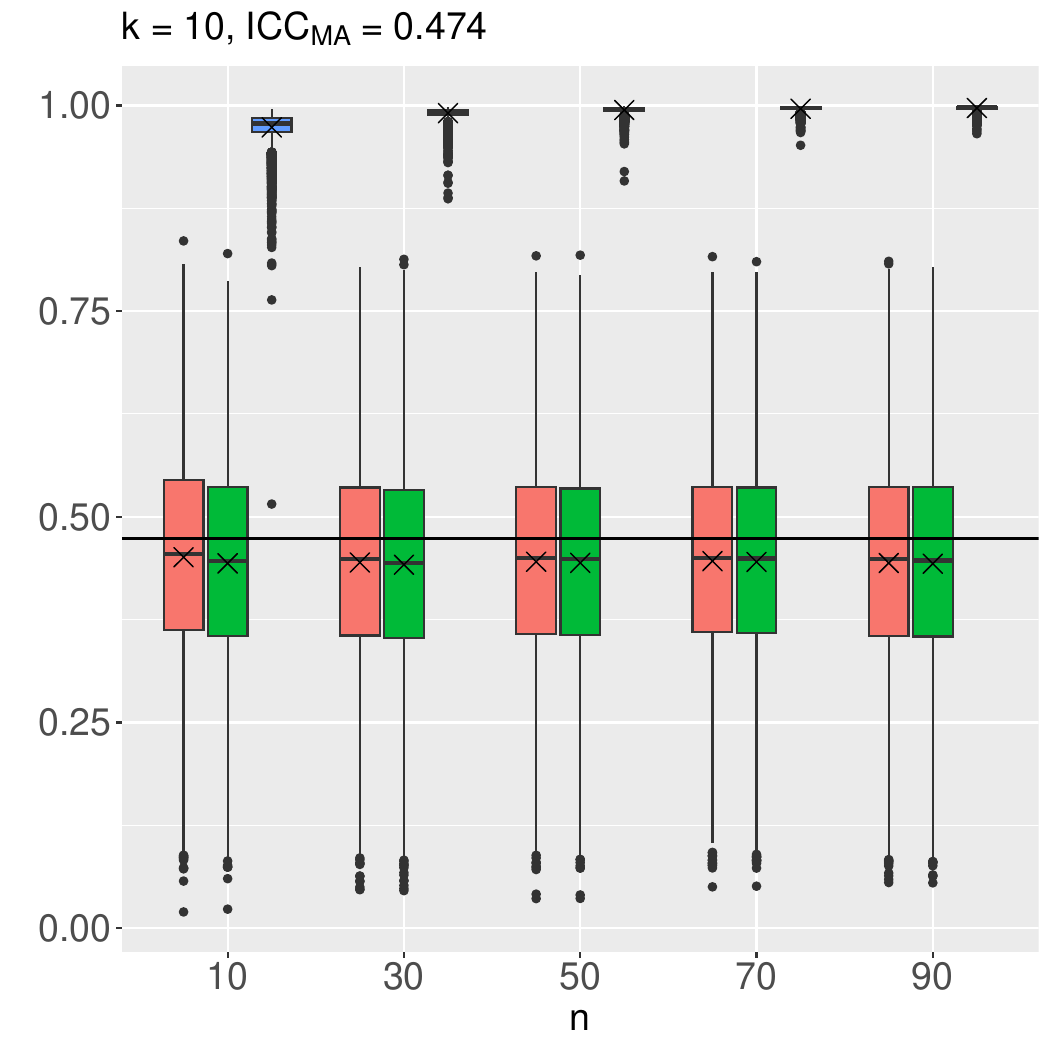,width=3.3in,angle=0}
		\end{tabular}{\caption{Boxplots of the three statistics for the mean with 10,000 repetitions. The red boxes represent the $I^2_{\rm A}$ statistic, the green boxes represent the $I^2_{\rm ANOVA}$ statistic, and the blue boxes represent the $I^2$ statistic. The crosses on each box are the mean values of the 10000 repetitions. The solid lines stand for the absolute heterogeneity ${\rm ICC}_{\rm MA}$.}\label{c1}}
	\end{center}
\end{figure}

Our next simulation is to examine the scenario where the common population variance assumption does not hold. For the $i$th study, we generate the population variance in each replication from a gamma distribution with shape parameter 25 and scale parameter 4, or equivalently, with mean 100 and variance 400. All other settings remain the same as in the previous simulation. Note that for this case, the true value of ${\rm ICC}_{\rm MA}$ will vary across replications because of the randomness in the population variance. Lastly, we present in Figure \ref{app} the simulation results together with an ${\rm ICC}_{\rm MA}$ value using $\sigma_{\rm pop}^2=100$, which, as can be seen, are similar to those for the common population variance.

\section{The $I^2_{\rm A}$ statistic for the mean difference}\label{sec5}
	
In addition to the mean considered in Section \ref{sec4}, two other commonly used effect sizes for continuous outcomes are the mean difference (MD) and the standardized mean difference (SMD). Following this, we will describe the $I^2_{\rm A}$ statistic for MD in this section and then for SMD in Section \ref{sec6}.
	
For a meta-analysis of MD, each study has two treatment arms including the treatment group and the control group. The summary statistics for each study then consist of the observed MD $y_i$, the sample sizes $n_i^T$ and $n_i^C$, and the standard errors $\hat\sigma_{y_i^T}$ and $\hat\sigma_{y_i^C}$ associated with the treatment and control groups.
Given these summary statistics, the estimated variance of the mean difference $y_i$ is $\hat\sigma_{y_i}^2=\hat\sigma_{y_i^T}^2+\hat\sigma_{y_i^C}^2$, which is also treated as the true within-study variance of $y_i$ as $\sigma_{y_i}^2=\sigma_{y_i^T}^2+\sigma_{y_i^C}^2$. In what follows, we describe the derivation procedure for the $I^2_{\rm A}$ statistic in the meta-analysis of MD, which extends from that for the meta-analysis of the mean. Following Section 4, we let $\sigma_{y_i}^2=(1/n_i^T+1/n_i^C)\sigma_{\rm pop}^2$, and moreover define the effective sample size for each study as $n_i=1/(1/n_i^T+1/n_i^C)$. This leads to the inverse-variance weights as $w_i=1/\sigma_{y_i}^2=n_i/\sigma_{\rm pop}^2$, and consequently, the $I^2_{\rm A}$ statistic for the meta-analysis of MD can again be expressed as
\beqr\label{iq12}
I_{\rm A}^2=\max\left\{\frac{Q-(k-1)}{Q+(k-1)(\tilde n-1)},0\right\}.
\eeqr
In other words, the newly derived $I_{\rm A}^2$ shares the same expression as in (\ref{iq1}) but with a different definition of $n_i$. 
Moreover, the expression in (\ref{unequalv}) also applies to the $I_{\rm A}^2$ statistic for MD. Just as property (d) applies to the $I_{\rm A}^2$ statistic for the mean, the $I_{\rm A}^2$ statistic for MD is also applicable and interpretable when the population variances differ. 
For more details on the model assumptions for MD, please refer to \hyperref[appD]{Appendix D}.
%In analogy to Theorem \ref{th1}, we give the property of the $I^2_{\rm Q}$ in this scenario in the following theorem.
%\begin{theorem}\label{th2}
%	For a meta-analysis with the mean difference as the effect sizes, assume that the sample sizes for both the treatment group and the control group satisfy $n_i^T,n_i^C\ge 2$ for all studies. Then, the relationship between the statistic for the absolute measure of heterogeneity, $I_{\rm Q}^2$ in (\ref{iq1}), and the statistic for the relative measure of heterogeneity, $I^2$ in (\ref{i2}), is as follows:
%	\beqr\label{relat}
%	I_{\rm Q}^2\le  I^2.
%	\eeqr
%	Equality holds when $Q-(k-1)\le 0$ or $n_i^T=n_i^C=2$ for all studies.
%\end{theorem}
%The establishment of Theorem \ref{th2} is straightforward. When $n_i^T,n_i^C\ge 2$, it follows that $n_i=1/(1/n_i^T + 1/n_i^C)\ge 1$. Equality holds if and only if $n_i^T=n_i^C=2$. Therefore, based on Theorem \ref{th1}, Theorem \ref{th2} is established.}

\subsection{Real data analysis}\label{sec5.1}
To exemplify the $I^2_{\rm A}$ statistic for MD, we revisit a meta-analysis conducted in Avery et al. (2022) \cite{avery2022efficacy}. This study explores the effect of interventions to taper long term opioid treatment for chronic non-cancer pain. Among the several interventions, we consider the effect of acupuncture. For each study, the observed effect size is the mean difference of reduced  opioid dose. For easy reference, we provide the summary data for the three studies in Table \ref{data2}.

By Table \ref{data2}, the estimated effect sizes $y_i$ for the three studies are $(32.0, -4.8, -14.8)$ and the within-study variances of $y_i$ are $(272.14,20.29,65.48)$, yielding that $\sum_{i=1}^3w_i=0.07$ and $\sum_{i=1}^3w_iy_i=0.35$. Moreover, Cochran's $Q$ statistic is given as $Q=6.50$. Further by formula (\ref{i2}), we have
\beqrs
I^2=\max\left\{\frac{6.50-(3-1)}{6.50},0\right\}=0.69.
\eeqrs
To compute the $I^2_{\rm A}$ statistic, the effective sample sizes $n_i$ for the three studies can be derived as 3.60, 26.67 and 8.74, respectively. This leads to the adjusted mean sample size as $\tilde n=9.24$, and moreover by  formula (\ref{iq1}),
\beqrs
I_{\rm A}^2=\max\left\{\frac{6.50-(3-1)}{6.50+(3-1)(9.24-1)},0\right\}=0.20.
\eeqrs 
Lastly, noting that $\bar y=-3.65$, ${\rm MSB}_{\rm MA}=2848.76$ and ${\rm MSW}_{\rm MA}=586.93$, we apply formula (\ref{iqb}) and it yields that
\beqrs
I^2_{\rm ANOVA}=\max\left\{\frac{2848.76-586.93}{2848.76+(9.24-1)\times 586.93},0\right\}=0.29.\eeqrs
To conclude, it is again evident that the values of $I^2_{\rm A}$ and $I^2_{\rm ANOVA}$ are close to each other, and both of them are significantly different from the value of $I^2$.

To further compare the $I^2_{\rm A}$, $I^2_{\rm ANOVA}$ and $I^2$ statistics, we also plot the population distributions for the three studies and the sampling distributions of the observed effect sizes in Figure \ref{figmd} for visualization. We note that two of the populations are largely overlapped with little heterogeneity, whereas the third population is moderately deviated. Given this, we conclude that the heterogeneity between the three studies may not be substantial overall, if measured by the $I^2_{\rm A}$ statistic. By contrast, the $I^2$ statistic concludes a very substantial heterogeneity between the sampling distributions of the observed effect sizes.
\begin{figure}[htp!]
	\begin{center}
		\begin{tabular}{cc}
			\psfig{figure=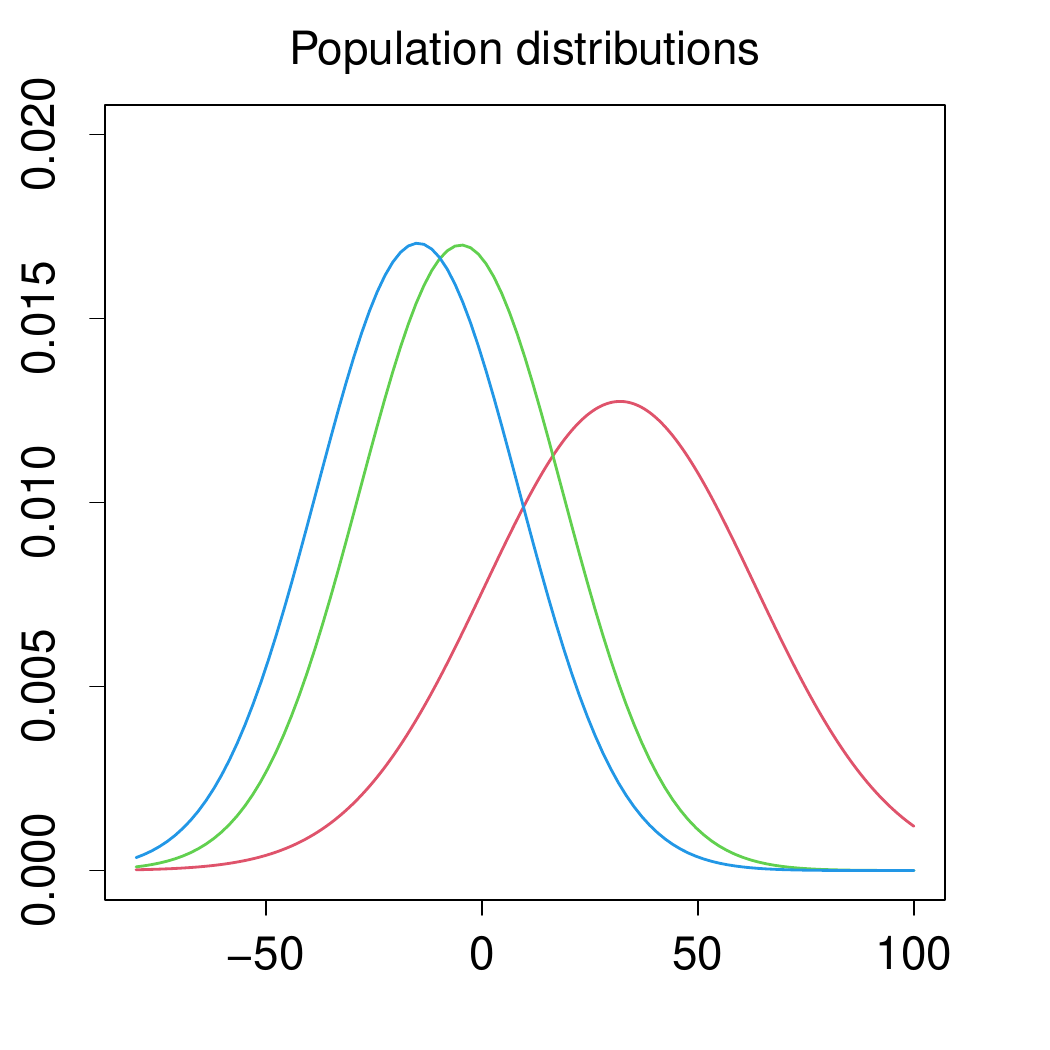,width=2.8in,angle=0}&
			\psfig{figure=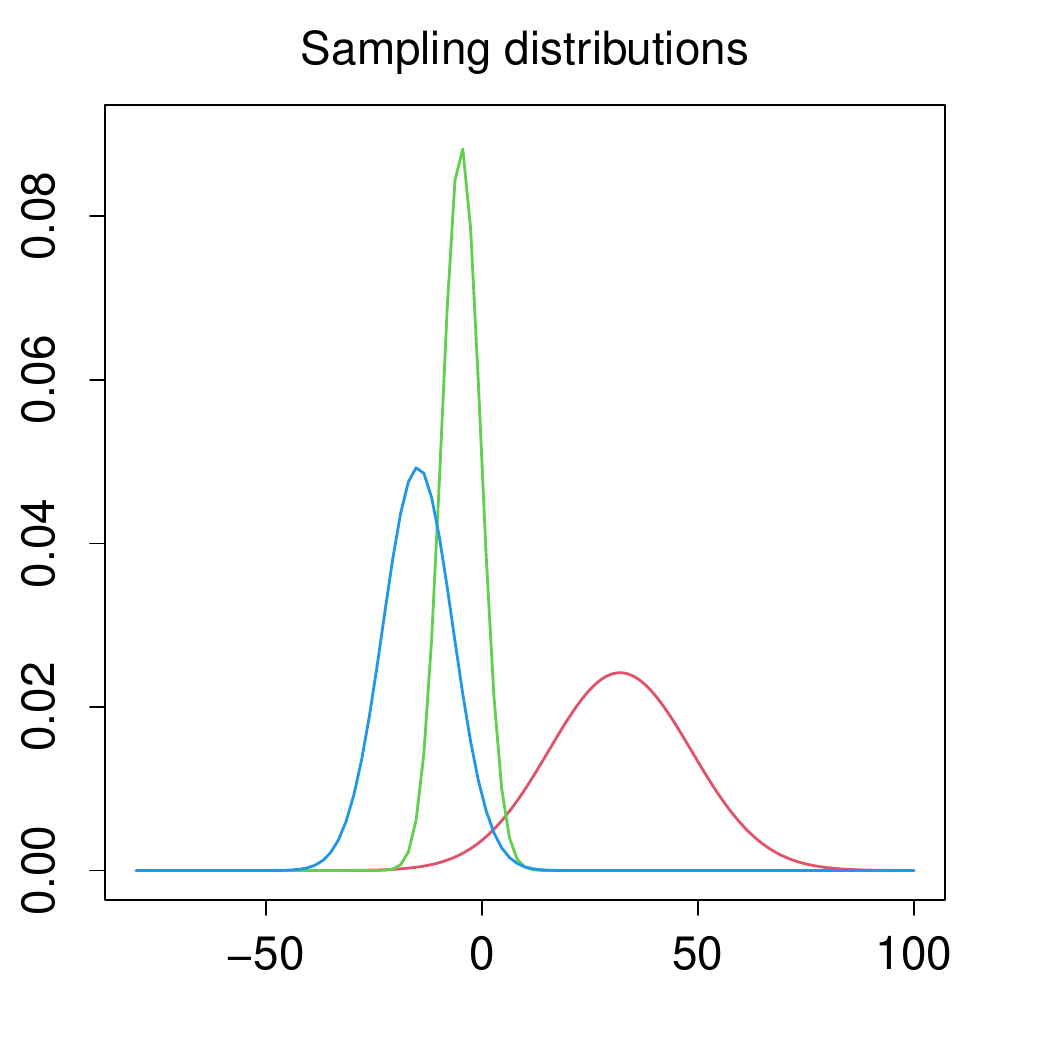,width=2.8in,angle=0}
		\end{tabular}{\caption{Population distributions of the three studies and the sampling distributions of the observed effect sizes with blue for Zheng (2008), green for Zheng (2019), and red for Jackson (2021). For each study, the population distribution is assumed to be normal with mean $y_i^T-y_i^C$ and variance $\{n_i^T(n_i^T-1)\hat\sigma_{y_i^T}^2+n_i^C(n_i^C-1)\hat\sigma_{y_i^C}^2\}/(n_i^T+n_i^C-2)$. The sampling distribution of the effect size is assumed to be normal with mean $y_i^T-y_i^C$ and variance $\hat\sigma_{y_i^T}^2+\hat\sigma_{y_i^C}^2$.}\label{figmd}}
	\end{center}
\end{figure}

\subsection{Numerical results}\label{sec5.2}
To numerically compare the $I^2_{\rm A}$, $I^2_{\rm ANOVA}$ and $I^2$ statistics, we generate the data from two-arm studies as follows:
\beqr\label{ipdmd}
\begin{aligned}
	y_{ij}^T&=\mu^T+\delta_i^T+\xi_{ij}^T,\quad j=1,\ldots,n_i^T,\\
	y_{ij'}^C&=\mu^C+\delta_i^C+\xi_{ij'}^C,\quad j'=1,\ldots,n_i^C,
\end{aligned}
\eeqr
where $\xi_{ij}^T$ and $\xi_{ij'}^C$ are i.i.d. normal random errors with mean 0 and common variance $\sigma^2$.
For a more detailed description of model (\ref{ipdmd}), one may refer to \hyperref[appD]{Appendix D}.

Without loss of generality, we set $\mu^T=\mu^C=0$ and $\sigma^2=1$. We also generate $\delta_i^T$ and $\delta_i^C$ independently from $N(0,0.045)$ or $N(0,0.45)$. With the observed effect sizes being $\sum_{j=1}^{n^T}y_{ij}^T/n_i^T-\sum_{j'=1}^{n^C}y_{ij'}^C/n_i^C$, the between-study variance is $\tau^2=0.09$ or $0.9$, yielding an ${\rm ICC}_{\rm MA}$ value of $0.083$ or $0.474$, respectively. For other settings, we consider $k=3$ or 10 to represent a small or large number of studies within the meta-analysis, and the sample sizes of both treatment arms, $n_i^T$ and $n_i^C$, to be identical. We further let the sample sizes for both arms of the $i$th study be $i*n$, where $i$ ranges from 1 to $k$, and $n$ varies from 10 to 90.
Then for each simulation setting, we proceed to generate the raw data and compute the summary statistics, including $y_i^T$, $y_i^C$, $\hat\sigma_{y_i^T}^2$ and $\hat\sigma_{y_i^C}^2$, for each of the $k$ studies. Finally with $M=10,000$ repetitions, we present the boxplots of the $I^2_{\rm A}$, $I^2_{\rm ANOVA}$ and $I^2$ statistics and also visualize their mean values in Figure \ref{c2}. Based on the numerical results, it is clear again that the $I^2$ statistic monotonically increases with the sample size $n$, whereas the $I^2_{\rm A}$ and $I^2_{\rm ANOVA}$ statistics are not affected by the sample size. Moreover, the two new statistics also yield similar estimates for ${\rm ICC}_{\rm MA}$ in most settings, as well as provide a better performance when the number of studies $k$ increases.
\begin{figure}[htp!]
	\begin{center}
		\begin{tabular}{cc}
			\psfig{figure=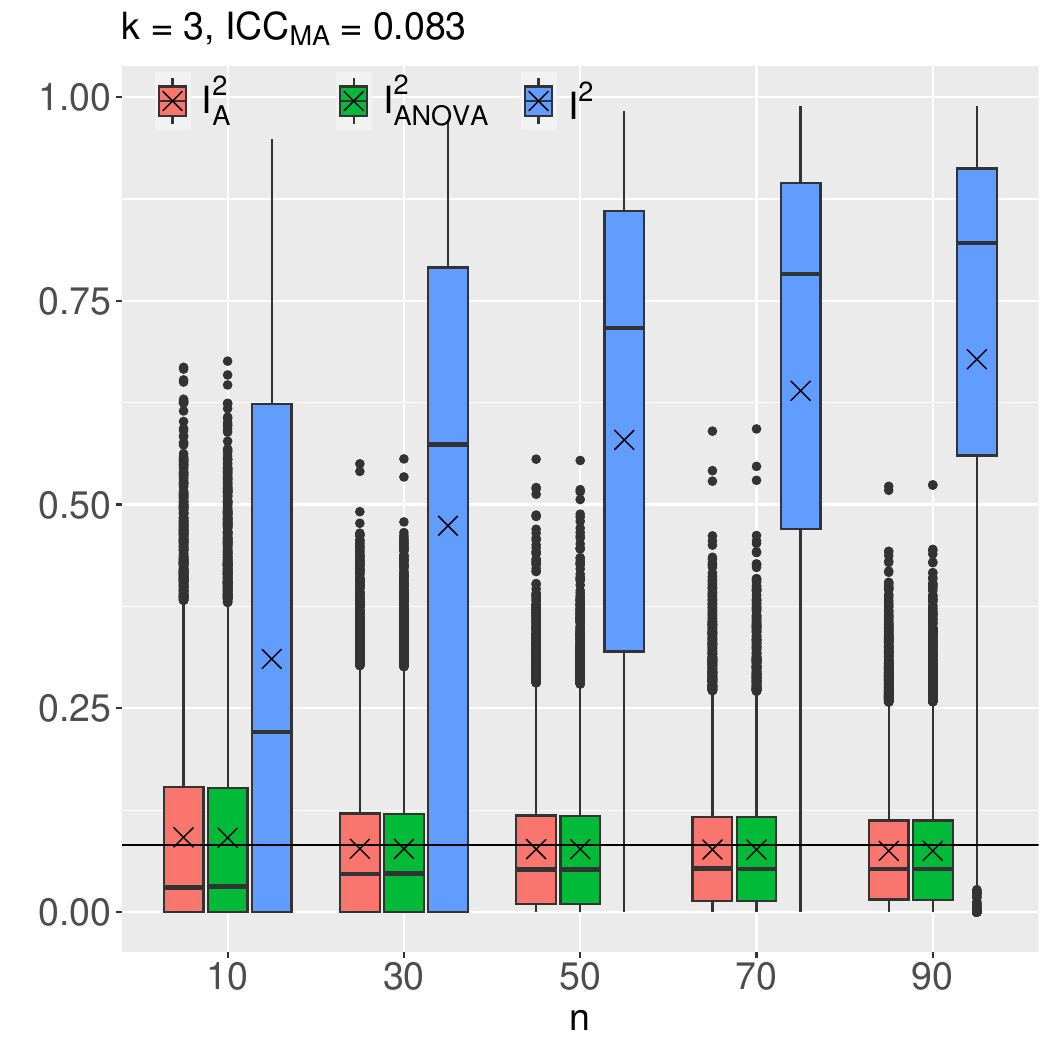,width=3.3in,angle=0}&
            \psfig{figure=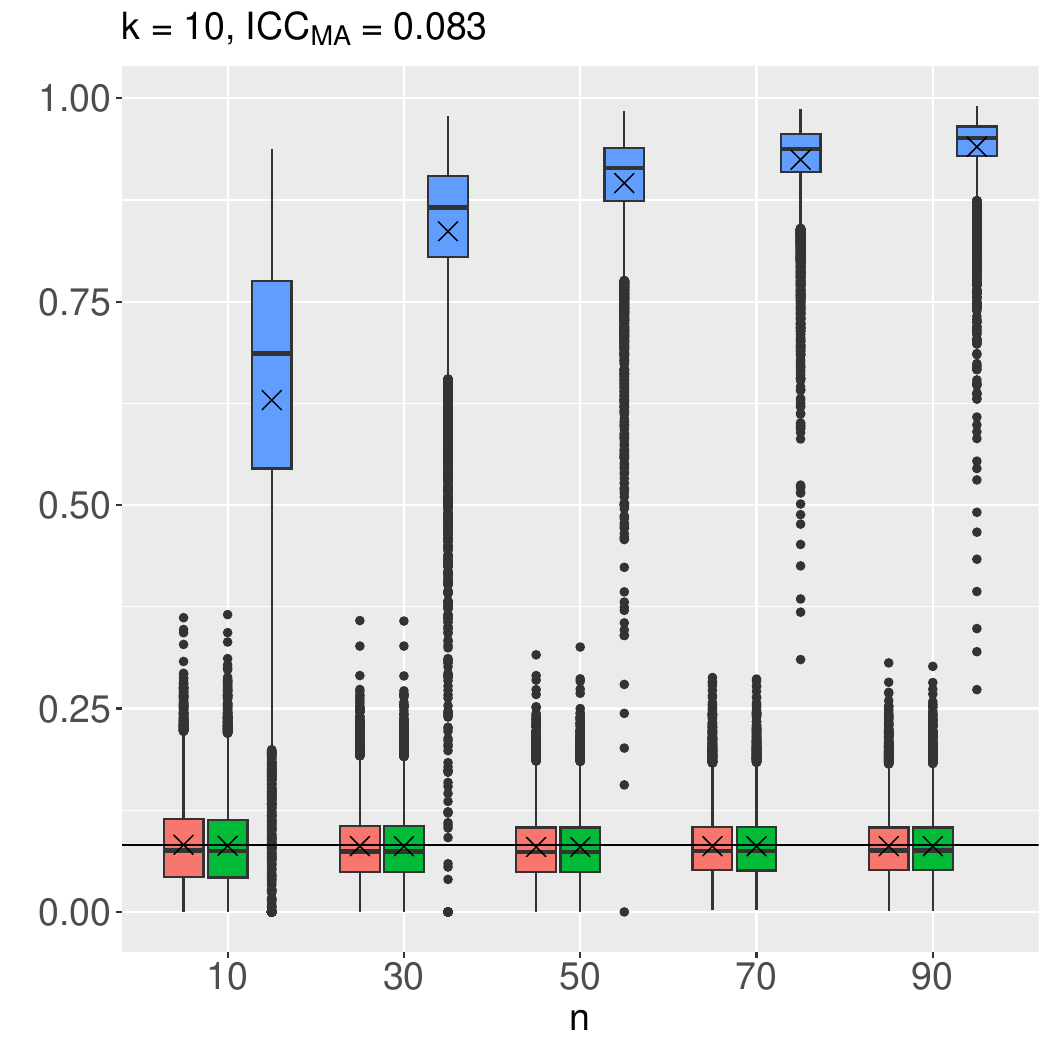,width=3.3in,angle=0}\\
            \psfig{figure=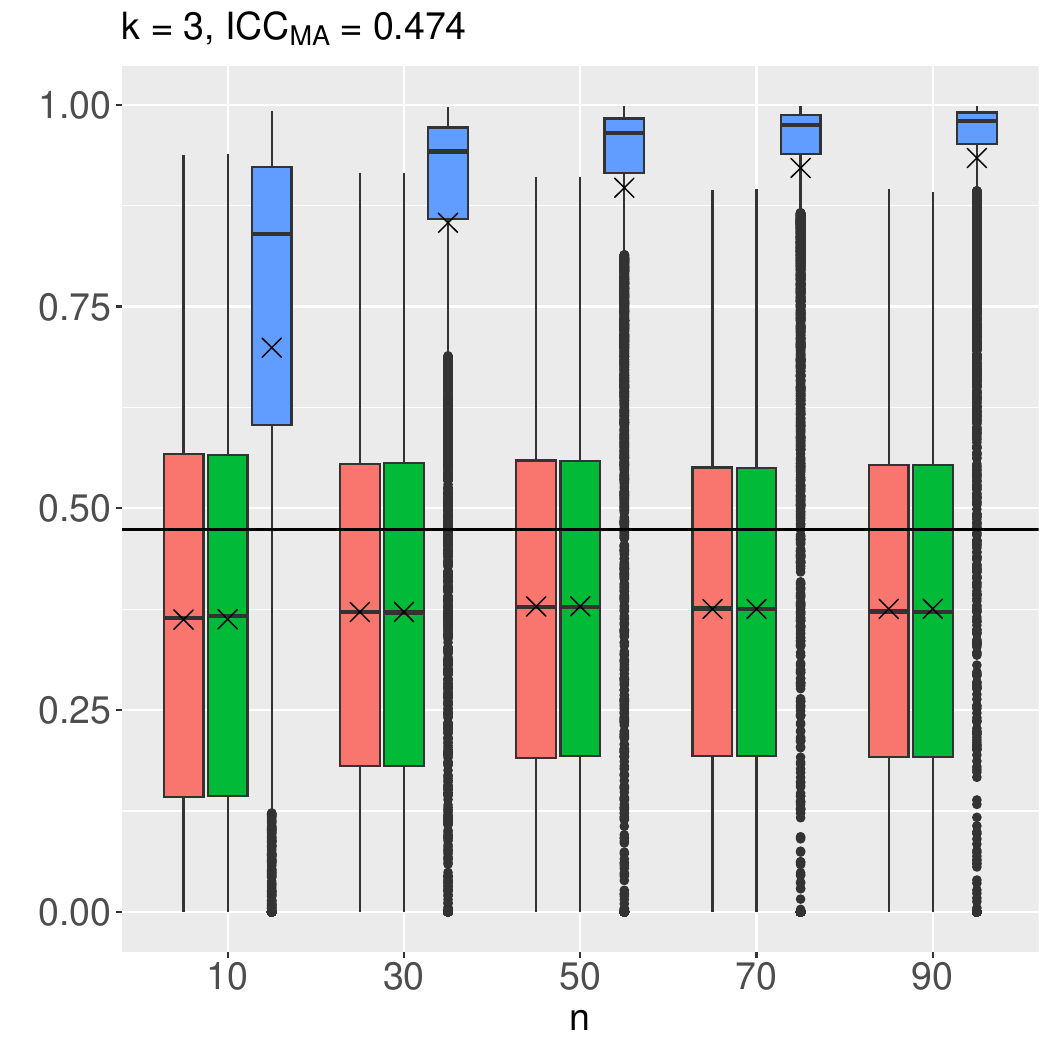,width=3.3in,angle=0}&
            \psfig{figure=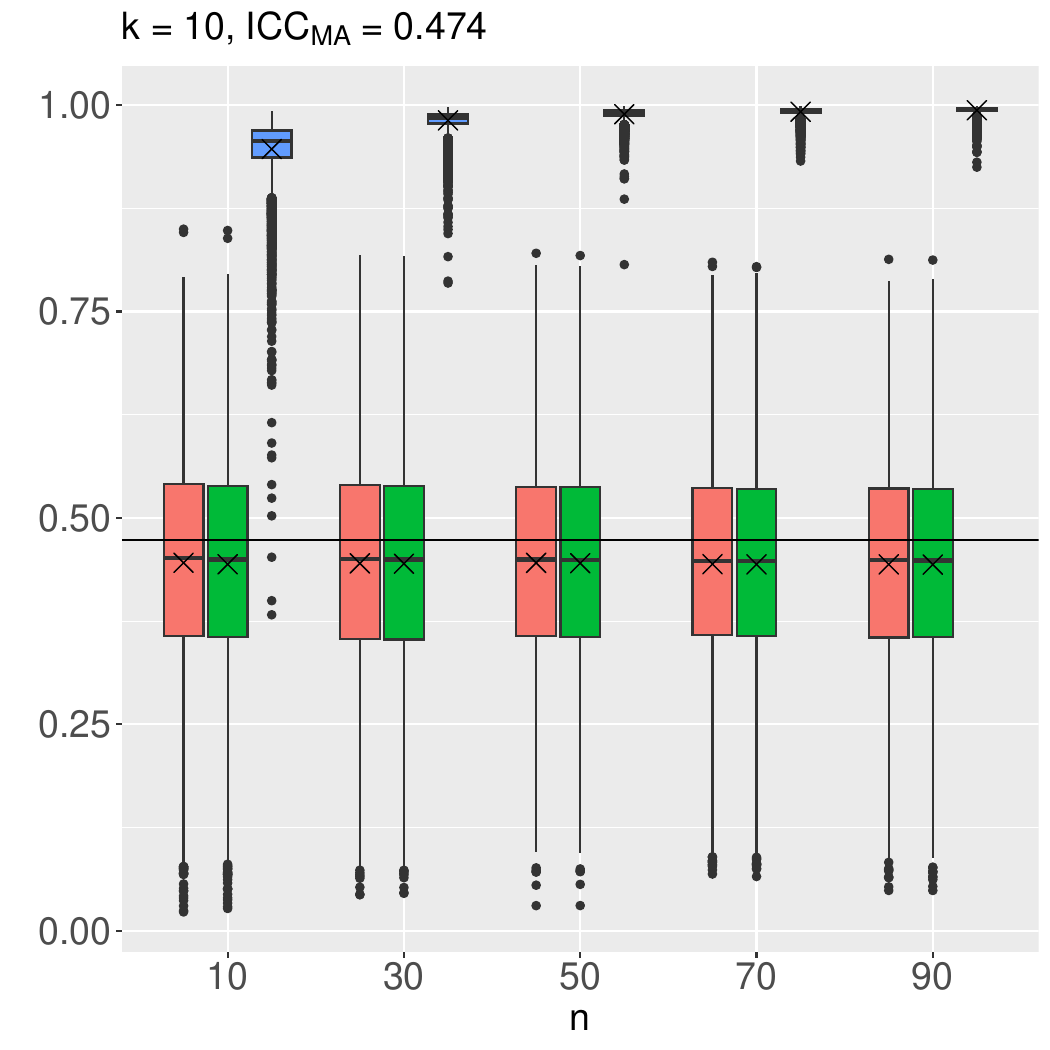,width=3.3in,angle=0}
		\end{tabular}{\caption{Boxplots of the three statistics for the mean difference with 10,000 repetitions. The red boxes represent the $I^2_{\rm A}$ statistic, the green boxes represent the $I^2_{\rm ANOVA}$ statistic, and the blue boxes represent the $I^2$ statistic. The crosses on each box are the mean values of the 10000 repetitions. The solid lines stand for the absolute heterogeneity ${\rm ICC}_{\rm MA}$.}\label{c2}}
	\end{center}
\end{figure}

\section{The $I^2_{\rm A}$ statistic for the standardized mean difference}\label{sec6}

In addition to the mean difference (MD), another commonly used effect size for continuous outcomes in two-arm studies is the standardized mean difference (SMD). SMD is particularly useful when the assumption of equal population variances across different studies cannot be made. In such cases, the mean difference in each study is standardized to a uniform scale, ensuring comparability for the subsequent meta-analysis. Consequently, the estimated standardized mean difference $y_i$ can be viewed as the observed mean difference of two population arms, both with a variance of 1, indicating $\sigma_{\rm pop}^2=1$.
For a comprehensive understanding of the model specifications, one may refer to \hyperref[appE]{Appendix E}. Lastly, to estimate ${\rm ICC}_{\rm MA}$ for SMD, we use the DerSimonian-Laird estimator to estimate $\tau^2$ and set $\sigma_{\rm pop}^2$ to 1 in formula (\ref{icc}), yielding the $I^2_{\rm A}$ statistic as
	\beqr\label{iq2}
	I^2_{\rm A}=\max\left\{\frac{Q-(k-1)}{Q+(k-1)(\tilde w-1)},0\right\},
	\eeqr
where $\tilde w=(\sum_{i=1}^kw_i-\sum_{i=1}^kw_i^2/\sum_{i=1}^kw_i)/(k-1)$.
%Here, $w_i$ also satisfy the assumption in Section \ref{sec5} that $w_i=1/\sigma_{y_i}^2=n_i/\sigma_{\rm pop}^2$, with $n_i=1/(1/n_i^T+1/n_i^C)$ being the adjusted sample sizes. Theoretically, $\sigma_{\rm pop}^2=1$ for SMD. Thus, $I^2_{\rm Q}$ in (\ref{iq2}) should be identical to that in (\ref{iq1}). Consequently, the property of $I^2_{\rm Q}$ for MD described in Theorem \ref{th2} should also hold for SMD. However, in practice, since $w_i$ are estimated values, Theorem \ref{th2} does not strictly hold for SMD.  

\subsection{Real data analysis}\label{sec6.1}
To assess the performance of the $I^2_{\rm A}$ statistic in quantifying the heterogeneity for SMD, we revisit the real data example presented in Section \ref{sec5.1}. With the summary data provided in 
Table \ref{data2}, we first compute the estimated SMD and its corresponding variance for each study. Two commonly used statistics for estimating SMD are Cohen's $d$ \cite{cohen2013statistical} and Hedges' $g$ \cite{hedges1981distribution}. For a detailed guide on computing Cohen's $d$ and Hedges' $g$, one may refer to Lin and Aloe (2021) \cite{lin2021evaluation}. In this section, we employ Hedges' $g$ that derives an unbiased estimate for SMD.

Following the formulas provided by Lin and Aloe (2021) \cite{lin2021evaluation}, we can derive the estimated SMDs for the three studies as $(0.96,-0.20, -0.62)$, and the within-study variances of $y_i$ as $(0.31,0.04, 0.12)$. Further by $\sum_{i=1}^3w_i=38.12$ and $\sum_{i=1}^3w_iy_i=-7.43$, Cochran's $Q$ statistic can be computed as $Q=5.83$. Thus by formula (\ref{i2}),
	\beqrs
	I^2=\max\left\{\frac{5.83-(3-1)}{5.83},0\right\}=0.66.
	\eeqrs
Noting also that $\sum_{i=1}^3w_i^2=784.06$ and $\tilde w=8.78$, by formula (\ref{iq2}) we have
\beqrs
I_{\rm A}^2=\max\left\{\frac{5.83-(3-1)}{5.83+(3-1)(8.78-1)},0\right\}=0.18.
\eeqrs 
Lastly, to compute the $I^2_{\rm ANOVA}$ statistic, we first derive the effective sample sizes $n_i$ for the three studies as 3.60, 26.67 and 8.74, respectively, Moreover, we have $\bar y=-0.19$, ${\rm MSB}_{\rm MA}=3.19$, and ${\rm MSW}_{\rm MA}=1$. Then by formula (\ref{iqb}),
\beqrs
I^2_{\rm ANOVA}=\max\left\{\frac{3.19-1}{3.19+(9.24-1)\times 1},0\right\}=0.19.
\eeqrs

To further compare the three statistics, we plot the scaled population distributions for the three studies and the sampling distributions of the observed effect sizes in Figure \ref{figsmd}. Specifically, with SMDs as the effect sizes, all the scaled populations have a common variance of 1. Moreover, we apply the estimated SMDs as the population means.
Compared to Figure \ref{figmd}, the three scaled populations in Figure \ref{figsmd} get more close to each other, resulting in smaller values for the $I^2_{\rm A}$ and $I^2_{\rm ANOVA}$ statistics. On the other hand, a measure of 0.66 for the $I^2$ statistic indicates a large heterogeneity between the observed effect sizes.
\begin{figure}[htp!]
	\begin{center}
		\begin{tabular}{cc}
			\psfig{figure=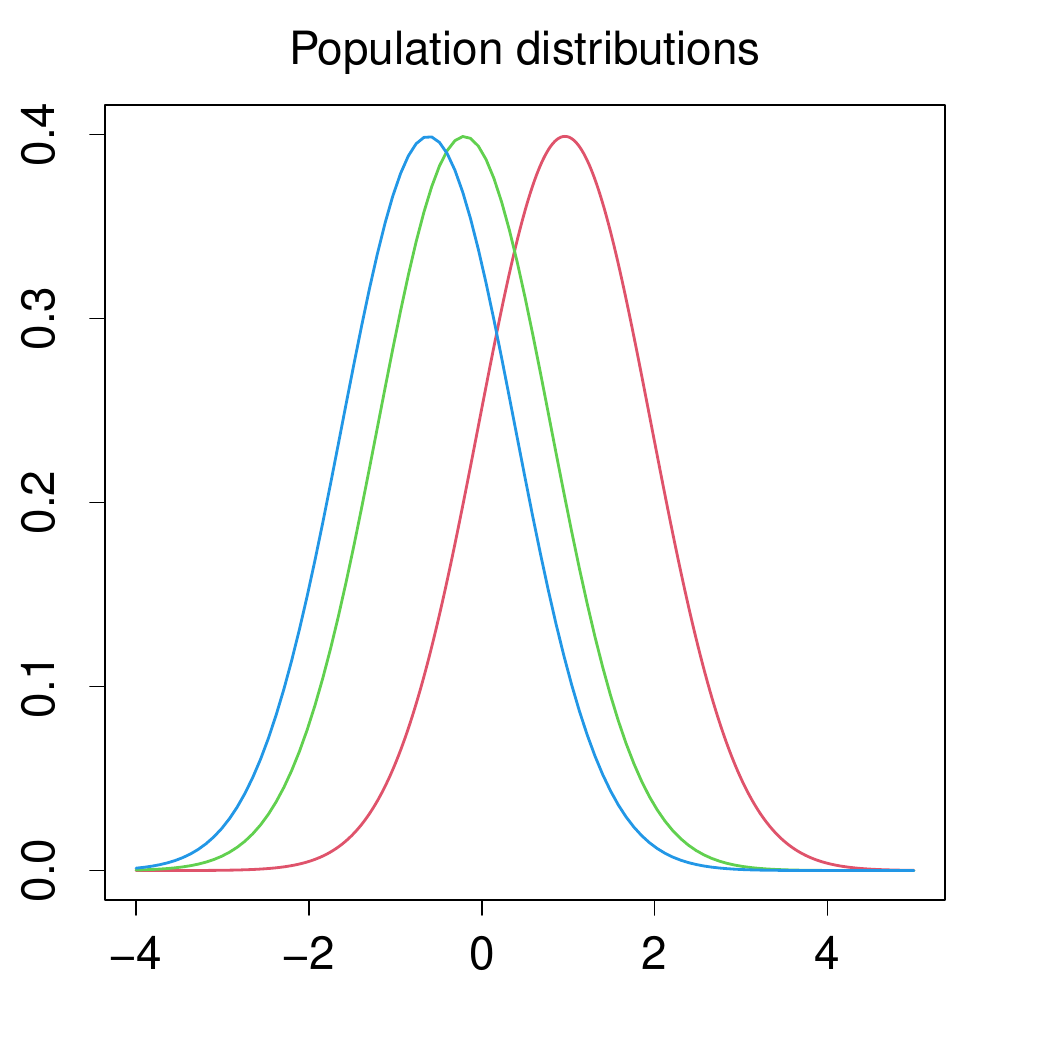,width=2.8in,angle=0}&
			\psfig{figure=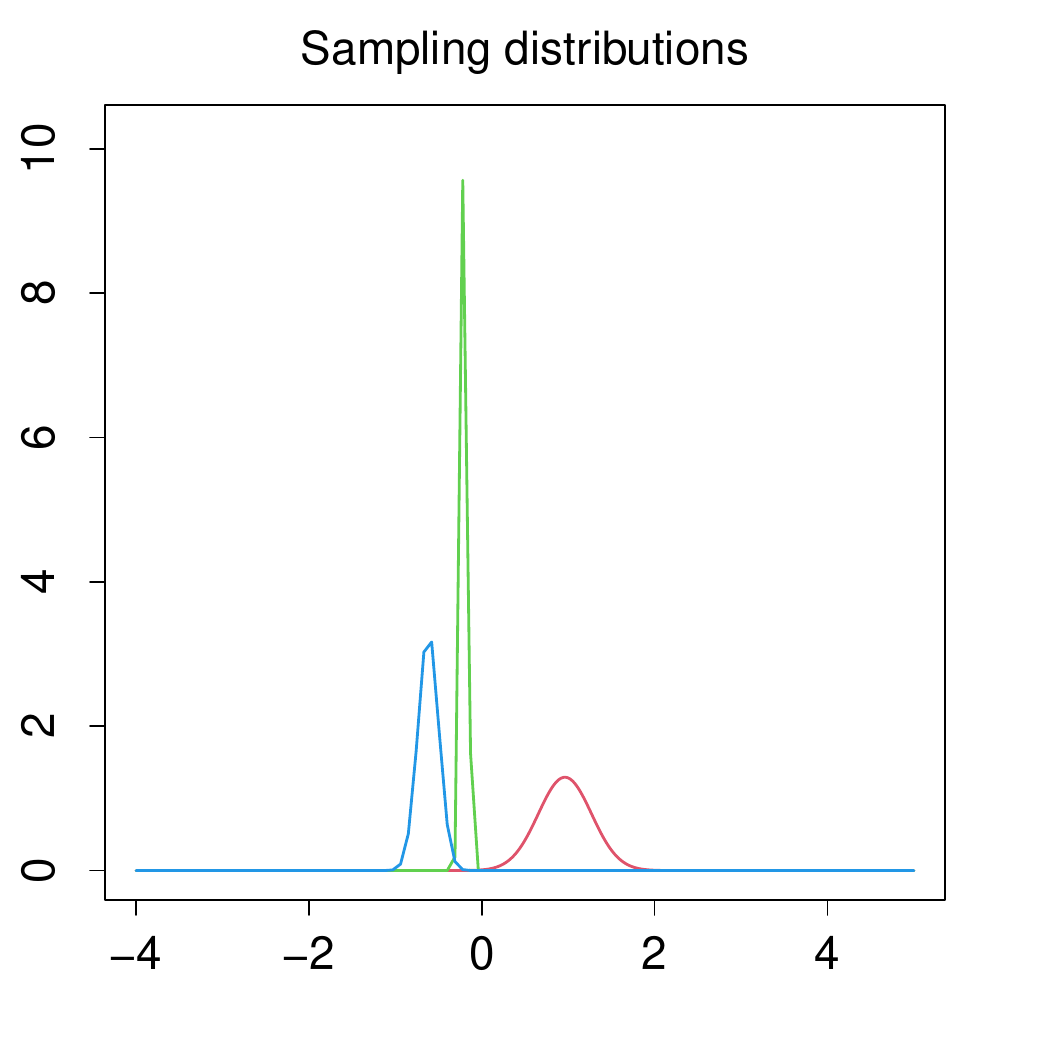,width=2.8in,angle=0}
		\end{tabular}{\caption{Population distributions of the three scaled studies and the sampling distributions of the observed effect sizes with blue for Zheng (2008), green for Zheng (2019), and red for Jackson (2021). For each study, the population distribution is assumed to be normal with mean SMD and variance 1. The sampling distribution of the effect size is assumed to be normal with mean SMD and the variance is assumed to be the within-study variance.}\label{figsmd}}
	\end{center}
\end{figure}

\subsection{Numerical results}\label{sec6.2}
To compare the $I^2_{\rm A}$, $I^2_{\rm ANOVA}$ and $I^2$ statistics for SMD, we generate the data from the following two-arm studies:
\beqr\label{ipdsmd}
\begin{aligned}	y_{ij}^T&=\sigma_i(\mu^T+\delta_i^T+\xi_{ij}^T),\quad j=1,\ldots,n_i^T,\\
	y_{ij'}^C&=\sigma_i(\mu^C+\delta_i^C+\xi_{ij'}^C),\quad j'=1,\ldots,n_i^C,
\end{aligned}
\eeqr
where $\xi_{ij}^T$ and $\xi_{ij'}^C$ are i.i.d. normal random errors with mean 0 and variance 1. Compared with model (\ref{ipdmd}), this new model contains an additional parameter $\sigma_i$, which is used to rescale each study. For a more detailed description of model (\ref{ipdsmd}), one may refer to \hyperref[appE]{Appendix E}.

In this simulation, we let $\sigma_i$ follow a uniform distribution ${\rm Unif}(0.5, 1.5)$, which yields unequal population variances for the $k$ studies and thus SMD ought to be applied rather than MD. The other settings are kept the same as those in Section \ref{sec6.2}. Then for each simulation setting, we proceed to generate the raw data and compute the summary statistics, including $y_i^T$, $y_i^C$, $\hat\sigma_{y_i^T}^2$ and $\hat\sigma_{y_i^C}^2$, for each of the $k$ studies. Finally with $M=10,000$ repetitions, we present the boxplots and the mean values of the $I^2_{\rm A}$, $I^2_{\rm ANOVA}$ and $I^2$ statistics in Figure \ref{c3}.
\begin{figure}[htp!]
	\begin{center}
		\begin{tabular}{cc}
			\psfig{figure=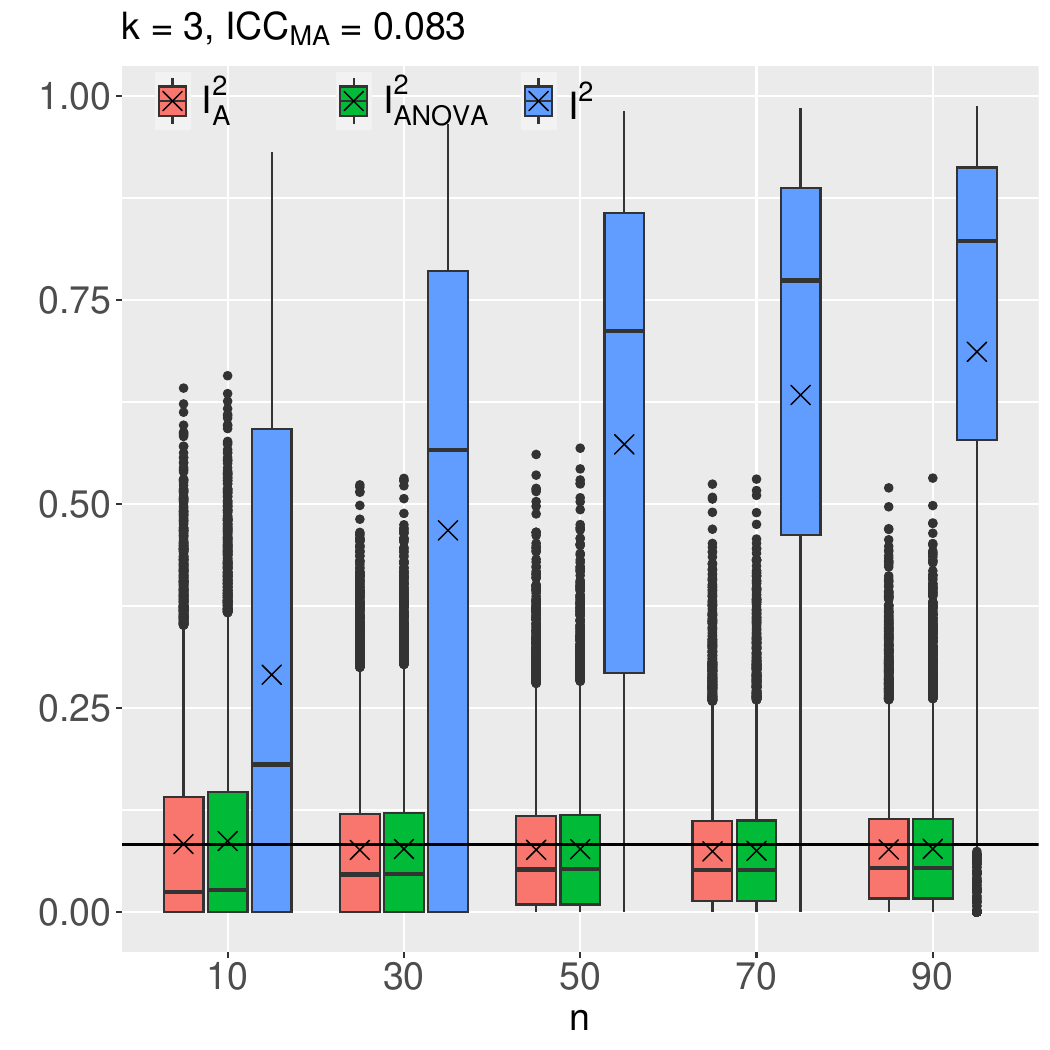,width=3.3in,angle=0}&
			\psfig{figure=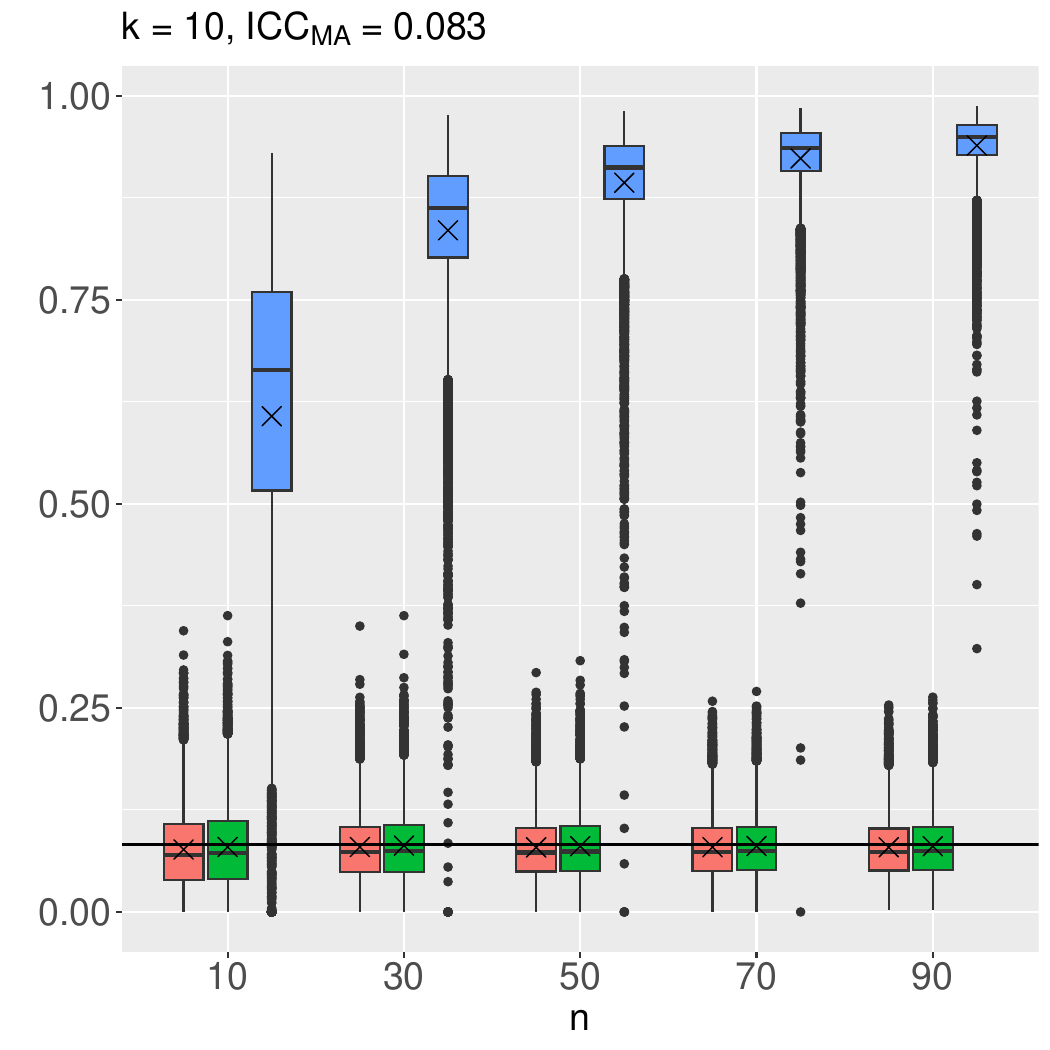,width=3.3in,angle=0}\\
			\psfig{figure=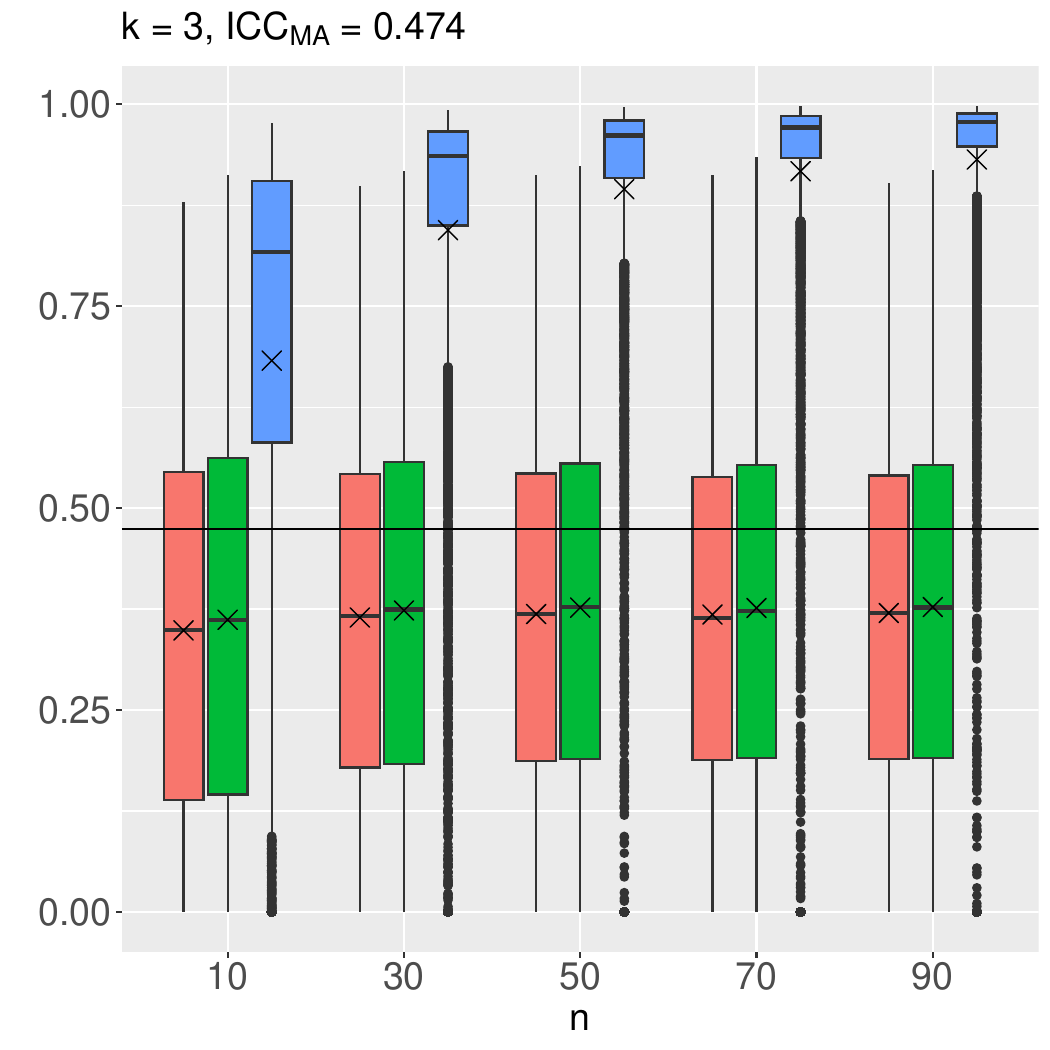,width=3.3in,angle=0}&
			\psfig{figure=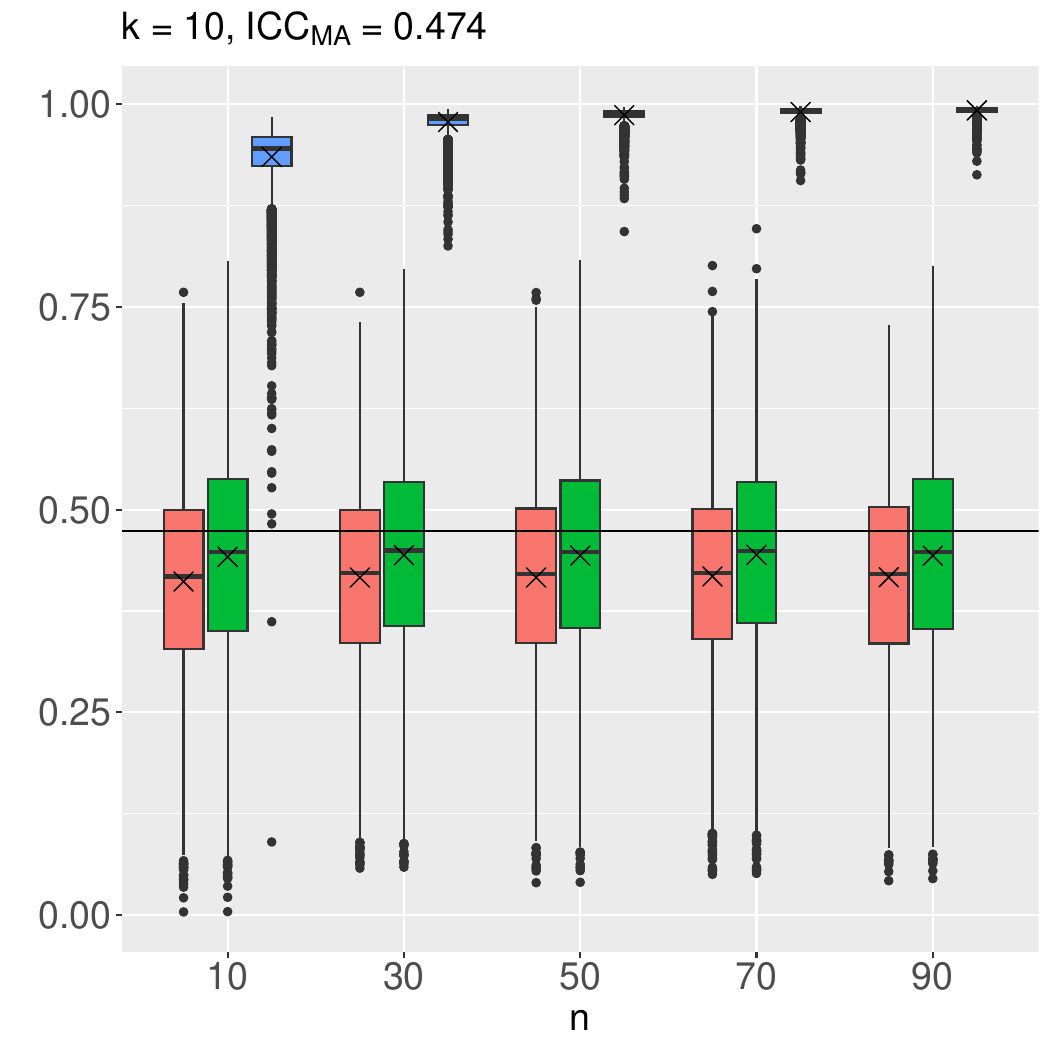,width=3.3in,angle=0}
		\end{tabular}{\caption{Boxplots of the three statistics for the standardized mean difference with 10,000 repetitions. The red boxes represent the $I^2_{\rm A}$ statistic, the green boxes represent the $I^2_{\rm ANOVA}$ statistic, and the blue boxes represent the $I^2$ statistic. The crosses on each box are the mean values of the 10000 repetitions. The solid lines stand for the absolute heterogeneity ${\rm ICC}_{\rm MA}$.}\label{c3}}
	\end{center}
\end{figure}

From Figure \ref{c3}, it is evident that the $I^2$ statistic is always monotonically increasing with the sample size $n$, which is consistent with the simulation results in Sections \ref{sec4.3} and \ref{sec5.2}. By contrast, the $I^2_{\rm A}$ and $I^2_{\rm ANOVA}$ statistics can always provide a good measure for the quantify of heterogeneity between the study populations, no matter whether the study sample sizes are large or not. For SMD, $I^2_{\rm ANOVA}$ provides a more accurate estimate for ${\rm ICC}_{\rm MA}$ with large heterogeneity compared with $I^2_{\rm A}$.

\section{Conclusion and discussion}\label{sec7}

Quantifying the heterogeneity is an important issue in meta-analysis for decision making. The presence of heterogeneity affects the extent to which generalizable conclusions can be formed and determines whether the random-effects model or the fixed-effect model should be employed. The $Q$ statistic is commonly used to test for the existence of the heterogeneity. However, as mentioned in the Cochrane Handbook for Systematic Reviews of Interventions \cite{higgins2019cochrane}, this test may have low power when the number of studies is small. Some also argue that the heterogeneity always exists, whether detectable by statistical tests or not. Thus, as a way to remedy, the $I^2$ statistic was further introduced to measure the extent of heterogeneity as
\beqrs
I^2
=\max\left\{\frac{Q-(k-1)}{Q},0\right\}.
\eeqrs
Nowadays, both the $Q$ statistic and the $I^2$ statistic are routinely reported in the forest plot in meta-analysis, and the choice between the random-effects model and the fixed-effect model often relies on these two statistics. More specifically, if the $p$-value of the $Q$ statistic is less than 0.1 and the $I^2$ statistic exceeds 0.5, the random-effects model is preferred for meta-analysis; otherwise, the fixed-effect model will be chosen \cite{jiang2021clinical,chinnaratha2016percutaneous,yang2012effect}. It is noted, however, that these two statistics are highly correlated since the $I^2$ statistic is a monotonically increasing function of the $Q$ statistic. Additionally, the $p$-value based on the $Q$ statistic only indicates whether there is a statistical significance\cite{gelman2006difference}, but not reflect regarding the biological difference between the studies. Even if heterogeneity is not statistically detected, it may still be clinically meaningful. Therefore, a random-effects model is often more appropriate, and the inclusion of prediction intervals is recommended in the practice of meta-analysis.

In this paper, we have introduced a new measure, denoted as ${\rm ICC}_{\rm MA}$, to quantify the between-study heterogeneity for meta-analysis. To explore the distinction between ${\rm ICC}_{\rm HT}$ and ${\rm ICC}_{\rm MA}$, we have also drawn an interesting connection between ANOVA and meta-analysis, and learned that the essence of ${\rm ICC}_{\rm HT}$ is to quantify the heterogeneity between the observed effect sizes. As demonstrated by the motivating example in Section \ref{sec2}, the sampling distributions of the observed effect sizes may exhibit a significant dependency on the sample sizes, and they will asymptotically converge to their true effect sizes. Accordingly, with large sample sizes, the observed effect sizes will also yield an increased ${\rm ICC}_{\rm HT}$ close to one, no matter whether the underlying heterogeneity between the study populations is truly large or not. As an important alternative, our newly defined ${\rm ICC}_{\rm MA}$ is proposed to directly quantify the heterogeneity between the study populations. More specifically, we have systematically studied the statistical properties of ${\rm ICC}_{\rm MA}$, including the monotonicity, the location and scale invariance, the study size invariance, and the sample size invariance. It is the sample size invariance that distinguishes our new absolute measure of heterogeneity from ${\rm ICC}_{\rm HT}$. 

Moreover, we have also proposed two new statistics to serve as the estimates of ${\rm ICC}_{\rm MA}$, where the $I^2_{\rm ANOVA}$ statistic in (\ref{iqb}) is derived based on ANOVA, and the $I^2_{\rm A}$ statistic in (\ref{iq1}) is derived based on the $Q$ statistic as
\beqrs
I_{\rm A}^2=\max\left\{\frac{Q-(k-1)}{Q+(k-1)(\tilde n-1)},0\right\},
\eeqrs	
where $\tilde n$ is the adjusted mean sample size for the $k$ studies. In addition, the $I_{\rm A}^2$ statistic can also be expressed as
$I_{\rm A}^2={\hat\tau^2} / (\hat\tau^2+\tilde n\tilde\sigma_y^2)$.
When the assumption of a common population variance holds, $\tilde n\tilde\sigma_y^2$ equals the common population variance $\sigma_{\rm pop}^2$; otherwise, it serves as a representative value for $\sigma_{\rm pop}^2$. This demonstrates that our $I_{\rm A}^2$ statistic is also widely applicable to scenarios where the population variances differ.
For practical use, the exact formulas for the $I^2_{\rm A}$ and $I^2_{\rm ANOVA}$ statistics are also derived under two other common scenarios with the mean difference or the standardized mean difference as the effect size. Simulations and real data analyses demonstrate that they both provide asymptotically unbiased estimators of the absolute heterogeneity between the study populations, and as expected, they also do not depend on the study sample sizes. For most cases, $I^2_{\rm A}$ and $I^2_{\rm ANOVA}$ show similar performance in estimating ${\rm ICC}_{\rm MA}$. However, for meta-analysis of the standardized mean difference with large heterogeneity, $I^2_{\rm ANOVA}$ offers a slightly better estimate of ${\rm ICC}_{\rm MA}$ than $I^2_{\rm A}$. Given that the $Q$ statistic is commonly reported in meta-analysis and that $I^2_{\rm A}$ can be conveniently calculated from the $Q$ statistic, we recommend using $I^2_{\rm A}$ in practical applications. But, of course, in case a higher accuracy is desired, the more complex $I^2_{\rm ANOVA}$ should be used, particularly for meta-analysis of the standardized mean difference.
To conclude, the $I^2_{\rm A}$ statistic can serve as a supplemental measure to monitor the situations where the study effect sizes are indeed similar with little biological difference. In such scenario, the fixed-effect model can be appropriate. Whereas if the sample sizes are very large, we note that the $I^2$ statistic may still rapidly increase to 1 showing a large heterogeneity and subsequently a random-effects model will continue to be adopted. In view of this, we are thus confident that the $I^2_{\rm A}$ statistic can add new value to meta-analysis, for example, being included in the forest plot as a supplement to the $I^2$ statistic.

In addition, as seen in Figures \ref{c1}, \ref{c2} and \ref{c3}, the $I^2_{\rm A}$ statistic tends to slightly underestimate ${\rm ICC}_{\rm MA}$ when $k$ is small and ${\rm ICC}_{\rm MA}$ is large. This underestimation may be primarily due to the inaccurate estimation of $\tau^2$. While the $I^2_{\rm A}$ statistic has the advantage of being directly expressed using the $Q$ statistic, making it more convenient to use, it implicitly relies on the DerSimonian-Laird (DL) method for estimating $\tau^2$. Although the DL estimator is most commonly used, it does have limitations, and numerous alternative methods have been proposed to further improve it, as summarized in Veroniki et al. (2016).\cite{veroniki2016methods} More recently,
	Kulinskaya et al. (2021)\cite{kulinskaya2021aq} and Bakbergenuly et al. (2022) \cite{bakbergenuly2022q} further highlighted that the $Q$ statistic may perform poorly in estimating $\tau^2$, partly because it does not account for the uncertainty in the within-study variances. To conclude, future research is warranted to investigate the impact of the $\tau^2$ estimate across various effect sizes on the estimation accuracy of ${\rm ICC}_{\rm MA}$, offering a broader range of options for estimating the measure of heterogeneity.

Lastly, it is worth noting that there are also several interesting directions for future research. First, the current work has presented its primary focus on meta-analysis with continuous outcomes. As a parallel work, it can be equally important for the $I^2_{\rm A}$ statistic to be further extended to meta-analysis with binary outcomes, which are also commonly encountered in clinical studies. 
However, extending $I^2_{\rm A}$ to meta-analysis with binary outcomes may not be straightforward. Unlike the continuous outcomes, the binary outcomes do not have a direct definition of the population variance for each study, making the task more complex. To illustrate this challenge, we now consider a study with two treatment arms and use the log odds ratio (lnOR) as the effect size. For the treatment group, let $a_i$ represent the number of events, $p_i^T$ the event rate, and $n_i^T$ the total number of participants. In this case, the event number $a_i$ follows the binomial distribution ${\rm Bin}(n_i^T, p_i^T)$. Similarly, let $b_i$, $p_i^C$ and $n_i^C$ be the number of events, the event rate, and the total number of participants in the control group, and moreover $b_i \sim {\rm Bin}(n_i^C, p_i^C)$. Then the observed effect size (lnOR) for the study will be calculated as $y_i=\ln[\{a_i/(n_i^T-a_i)\}/\{b_i/(n_i^C-b_i)\}]$, with the within-study variance estimated by $\sigma_{y_i}^2=1/a_i+1/(n_i^T-a_i)+1/b_i+1/(n_i^C-b_i)$. For more details, one may refer to Higgins et al. (2019).\cite{higgins2019cochrane} Given this setup, it may not be straightforward to define a common study population variance $\sigma_{\rm pop}^2$ based on the within-study variance $\sigma_{y_i}^2$ and the sample sizes $n_i^T$ and $n_i^C$ and so requires further investigation.

\vspace{1cm}
\bmsection*{Acknowledgments}
The authors sincerely thank the Editor, the Associate Editor, and the two reviewers for their constructive comments that have led to a substantial improvement of this paper. 
Liping Zhu and Ke Yang's research was supported in part by the National Key R\&D Program of China (2022YFA1003702) and the National Natural Science Foundation of China (12371294). 
Wangli Xu's research was supported in part by the MOE Project of Key Research Institute of Humanities and Social Sciences (No. 22JJD910001).
Tiejun Tong's research was supported in part by the General Research Fund of Hong Kong (HKBU12303421 and HKBU12300123), the Initiation Grant for Faculty Niche Research Areas (RC-FNRA-IG/23-24/SCI/03) of Hong Kong Baptist University, and the National Natural Science Foundation of China (12071305).

%\bibliography{wileyNJD-AMA}

\newpage
\begin{table*}
	\begin{center}{\caption{Connection between the ANOVA model in (\ref{IPD}) and the meta-analysis model in (\ref{REM}), where $y_i=\sum_{j=1}^{n_i}y_{ij}/n_i$ and $\epsilon_i=\sum_{j=1}^{n_i}\xi_{ij}/n_i$ for $i=1,\dots,k$ and $j=1,\dots,n_i$.
			}\label{table1}}
		\begin{tabular*}{300pt}{@{\extracolsep\fill}l|cc@{\extracolsep\fill}}
			\hline
			& \mbox{ANOVA} & \mbox{Meta-analysis}\\
			\hline
			\mbox{Model} & $y_{ij}=\mu+\delta_i+\xi_{ij}$ & $y_i=\mu+\delta_i+\epsilon_i$\\
			\mbox{Between-study variance}& $\tau^2$&$\tau^2$\\
			\mbox{Error (or within-study) variance}&$\sigma^2$& $\sigma^2/n_i$\\
			\mbox{Total variance}& ${\rm var}(y_{ij})=\tau^2+\sigma^2$&${\rm var}(y_{i})=\tau^2+\sigma^2/n_i$\\
			\hline
		\end{tabular*}
	\end{center}
\end{table*}

\begin{table}
	\begin{center}
		{\caption{The summary data of the ten studies for the meta-analysis from Jeong et al. (2014), where $y_i$ are the observed effect sizes and $n_i$ are the study sample sizes.}\label{data}}
		\begin{tabular*}{250pt}{@{\extracolsep\fill}lrrr@{\extracolsep\fill}}
			\hline
			Study&$y_i$ & $n_i$ & $\hat\sigma_{y_i}^2$\\
			\hline
			Wang (2013)& -3.10& 8& 1.81\\
			Prasad (2012)& -6.30& 11& 3.16\\
			Moniche (2012)& -9.40& 10& 0.53\\
			Friedrich (2012)& -14.20& 20& 3.04\\
			Honmou (2011)& -7.00& 12& 1.40\\
			Savitz (2011)& -9.00& 10& 1.60\\
			Battistella (2011)& -3.40& 6& 2.41\\
			Suarez (2009)& -2.20& 5& 1.15\\
			Savitz (2005)& -1.40& 5& 0.97\\
			Bang (2005)& -2.00& 5& 1.06\\
			\hline
		\end{tabular*}
	\end{center}
\end{table}

\begin{table}
	\begin{center}
		{\caption{The summary data of the three studies for the meta-analysis from Avery et al. (2022).\hspace{10cm}}\label{data2}}
		\begin{tabular*}{350pt}{@{\extracolsep\fill}lrrrrrr@{\extracolsep\fill}}
			\hline
			Study&$y_i^T$ & $n_i^T$ & $\hat\sigma_{y_i^T}$&$y_i^C$ & $n_i^C$ & $\hat\sigma_{y_i^C}$\\
			\hline
			Jackson (2021)& -34& 9& 10.43& -66& 6& 12.78\\
			Zheng (2019)& -13.6& 48& 3.23& -8.8& 60& 3.14\\
			Zheng (2008)& -25.7& 17& 7.59& -10.9& 18& 2.80\\
			\hline
		\end{tabular*}
	\end{center}
\end{table}

\clearpage
\appendix
\renewcommand{\thesection}{\Alph{section}.}
\bmsection{Proof of the properties of ICC$_{\rm MA}$}\label{appA}
\vspace*{12pt}

\begin{proof}[Proof of \textit{``Monotonicity"}]
	
	By the definition in (\ref{icc}), we can rewrite ICC$_{\rm MA}$ as 
	\beqrs
	{\rm ICC}_{\rm MA}=\frac{1}{1+\sigma_{\rm pop}^2/\tau^2}.
	\eeqrs
	This shows that ICC$_{\rm MA}$ is a monotonically increasing function of $\tau^2/\sigma_{\rm pop}^2$ and so property (i$'$) holds. 
\end{proof}

\begin{proof}[Proof of \textit{``Location and scale invariance"}]
	
	To prove the location and scale invariance, for any constants $a$ and $b>0$, we assume that the newly observed effect sizes are $y'_{ij}=a+by_{ij}$ for $i=1,\ldots,k$ and $j=1,\ldots,n_i$. Let also $\mu'_i=a+b\mu_i$ be the true effect sizes of the new study populations. Then consequently, the between-study variance and the common population variance are given as
	\beqrs
	(\tau^2)'&=&{\rm var}(\mu'_i)={\rm var}(a+b\mu_i)=b^2\tau^2,\\
	(\sigma_{\rm pop}^2)'&=&{\rm var}(a+by_{ij}|a+b\mu_i)=b^2\sigma_{\rm pop}^2.
	\eeqrs
	Further by (\ref{icc}), the measure of heterogeneity between the new studies is
	\beqrs
	{\rm ICC}'_{\rm MA}=\frac{(\tau^2)'}{(\tau^2)'+(\sigma_{\rm pop}^2)'}=\frac{b^2\tau^2}{b^2\tau^2+b^2\sigma_{\rm pop}^2}=\frac{\tau^2}{\tau^2+\sigma_{\rm pop}^2}={\rm ICC}_{\rm MA}.
	\eeqrs
	This verifies the property of location and scale invariance.
\end{proof}
\begin{proof}[Proof of \textit{``Study size invariance"}]
	To prove the study size invariance, we assume there are a total of $k'$ studies. Then by the random-effects model in (\ref{meta}), since the individual means $\mu_i$ are i.i.d. from $N(\mu,\tau^2)$, the between-study variance will remain unchanged as $\tau^2$ regardless of the number of studies. Further by the common population variance assumption, we have ${\rm var}(y_{ij}|\mu_i)=\sigma_{\rm pop}^2$ for all $i=1,\dots,k'$ and $j=1,\dots,n_i$. This proves the property of study size invariance.
\end{proof}
\begin{proof}[Proof of \textit{``Sample size invariance"}]
	To prove the sample size invariance, we assume that the new sample sizes are $n_i'$ for each study, and consequently $y_i'=\sum_{j=1}^{n_i'} y_{ij}/n_i'$ are the new effect sizes. Then under the common population variance assumption that ${\rm var}(y_{ij}|\mu_i)=\sigma_{\rm pop}^2$ for all $i$ and $j$, we have $\sigma_{y_i'}^2 = {\rm var}(y_i'|\mu_i) = \sigma_{\rm pop}^2/n_i'$, or equivalently, $n_i'\sigma_{y_i'}^2=\sigma_{\rm pop}^2$.
	That is, no matter how the sample sizes vary, the common population variance will always remain unchanged. Finally, noting that $\tau^2$ also remains since the study populations are  unaltered, we thus have the property of sample size invariance.
\end{proof}

\bmsection{Methods for estimating ICC}\label{appB}

To estimate ICC from the random-effects ANOVA in (\ref{IPD}), we first partition the total variation of the observations into two components as
\beqr\label{ident}
\sum_{i=1}^k\sum_{j=1}^{n_i}(y_{ij}-\bar y)^2=\sum_{i=1}^kn_i(y_i-\bar y)^2+\sum_{i=1}^k\sum_{j=1}^{n_i}(y_{ij}-y_i)^2,
\eeqr
where $y_i=\sum_{j=1}^{n_i}y_{ij}/n_i$ are the individual sample means, and $\bar y=\sum_{i=1}^k\sum_{j=1}^{n_i}y_{ij}/\sum_{i=1}^kn_i$ is the grand sample mean. More specifically, the term on the left-hand side of (\ref{ident}) is the total sum of squares (SST), and the two terms on the right-hand side are the sum of squares between the populations (SSB) and the error sum of squares within the populations (SSW), respectively.

By equating SSB and SSW to their respective expected values,  Cochran (1939) \cite{cochran1939use} derived the method of moments estimators of $\tau^2$ and $\sigma^2$. Further by plugging these two estimators in formula (\ref{rho}), it yields the ANOVA estimator for the unknown ICC. By Smith (1957) \cite{smith1957estimation}, the ANOVA estimator is a biased but consistent estimator. Moreover, as the method of moments estimators may take a negative value when ${\rm SSB}/k<{\rm SSW}/(\sum_{i=1}^k(n_i-1)$, one often truncates the negative value to 0 when it occurs. For the balanced case when the sample sizes are all equal, Searle (1971) \cite{searle1971linear} derived an exact confidence interval for ICC based on the ANOVA table. For the unbalanced case, however, the exact confidence interval from the ANOVA table is not available. As a remedy, Thomas and Hultquist (1978) \cite{thomas1978interval} and Donner (1979) \cite{donner1979use} suggested an adjusted confidence interval in which the common sample size in the balanced case is replaced by the average sample size. They further showed by simulation studies that the adjusted confidence interval performs very well in terms of the coverage probability.

Besides the well-known ANOVA estimator, it is noteworthy that there are also other estimators for ICC in the literature. To name a few, Thomas and Hultquist (1978) \cite{thomas1978interval} constructed a confidence interval for ICC based on the unweighted average of the individual sample means $\tilde y=\sum_{i=1}^ky_i/k$. Observing that ${\rm ICC}=(\tau^2/\sigma^2)/(\tau^2/\sigma^2+1)$, Wald (1940) \cite{wald1940note} proposed another estimator for ICC by first estimating $\tau^2/\sigma^2$, yet as a limitation, there does not exist a closed form for either the point estimator or its confidence interval. As another alternative, by the facts that ${\rm cov}(y_{ij},y_{il})=\tau^2$ for $j\ne l$ and ${\rm var}(y_{ij})=\tau^2+\sigma^2$, Karlin
et al. (1981) \cite{karlin1981sibling} proposed to estimate ICC by the Pearson product-moment correlation computed over all the possible pairs of $(y_{ij},y_{il})$ for $j\ne l$ with some weighting schemes. In addition, Donner and Koval (1980a,b) \cite{donner1980estimation,donner1980large} proposed an iterative algorithm to compute the maximum likelihood estimator (MLE) for ICC directly, and presented its performance by simulations when the number of studies is large. For more estimators of ICC, one may also refer to Donner (1986) \cite{donner1986review}, Sahai and Ojeda (2004) \cite{sahai2004analysis}, and the references therein.

Despite the rich literature on the estimation of ICC, none of the existing estimators is known to be uniformly better than the others in the unbalanced case \cite{sahai2004analysis}. In practice, thanks to its simple and elegant form, the ANOVA estimator is frequently treated as the optimal estimator and so is most commonly used for estimating ICC. Lastly, we also note that the ANOVA estimator and the confidence interval suggested by Thomas and Hultquist (1978) \cite{thomas1978interval} and Donner (1979) \cite{donner1979use} can be readily implemented by the function \textit{ICCest} in the R package `ICC'.

\bmsection{The derivation of the $I^2_{\rm ANOVA}$ statistic}\label{appC}

To estimate ${\rm ICC}_{\rm MA}$, we begin by presenting the following lemma along with its proof.
\begin{lemma}\label{lemm1}
	With model (\ref{IPD}) and the summary data $y_i$, $\hat{\sigma}^2_{y_i}$ for $i=1,\dots,k$ in meta-analysis, we have $E({\rm MSB}_{\rm MA})=\tilde n\tau^2+\sigma_{\rm pop}^2$ and $E({\rm MSW}_{\rm MA})=\sigma_{\rm pop}^2$.
\end{lemma}

\begin{proof}[Proof of Lemma \ref{lemm1}]
	Denote by $\sigma_{y_i}^2=\sigma^2/n_i$. With the summary data, $y_i$ are independent normal random variables with mean $\mu$ and variances $\tau^2+\sigma_{y_i}^2$. Then the variance of $\sum_{i=1}^kn_iy_i$ is
	\beqrs
	{\rm Var}\left(\sum_{i=1}^kn_iy_i\right)=\sum_{i=1}^k{\rm Var}\left(n_iy_i\right)=\tau^2\sum_{i=1}^kn_i^2+\sum_{i=1}^kn_i^2\sigma_{y_i}^2.
	\eeqrs
	Thus,
	\beqrs
	E\left(\sum_{i=1}^kn_iy_i\right)^2&=&{\rm Var}\left(\sum_{i=1}^kn_iy_i\right)+\left\{E\left(\sum_{i=1}^kn_iy_i\right)\right\}^2\\
	&=&\tau^2\sum_{i=1}^kn_i^2+\sum_{i=1}^kn_i^2\sigma_{y_i}^2+\mu^2\left(\sum_{i=1}^kn_i\right)^2.
	\eeqrs
	Further, it can be derived that
	\beqrs
	E\left\{\sum_{i=1}^kn_i\left(y_i-\bar y\right)^2\right\}&=&\sum_{i=1}^kn_iE\left(y_i^2\right)-\frac{1}{\sum_{i=1}^kn_i}E\left(\sum_{i=1}^kn_iy_i\right)^2\\
	&=&\sum_{i=1}^{k}n_i\left[{\rm Var}\left(y_i\right)+\left\{E\left(y_i\right)\right\}^2\right]-\frac{1}{\sum_{i=1}^kn_i}E\left(\sum_{i=1}^kn_iy_i\right)^2\\
	&=&\sum_{i=1}^kn_i\left(\tau^2+\sigma_{y_i}^2+\mu^2\right)-\frac{1}{\sum_{i=1}^kn_i}\left\{\tau^2\sum_{i=1}^kn_i^2+\sum_{i=1}^kn_i^2\sigma_{y_i}^2+\mu^2(\sum_{i=1}^kn_i)^2\right\}\\
	&=&\tau^2(\sum_{i=1}^kn_i-\frac{\sum_{i=1}^kn^2_i}{\sum_{i=1}^kn_i})+\sum_{i=1}^kn_i\sigma^2_{y_i}-\frac{\sum_{i=1}^kn^2_i\sigma^2_{y_i}}{\sum_{i=1}^kn_i}.
	\eeqrs
	Since $\sigma_{y_i}^2=\sigma_{\rm pop}^2/n_i$, and $\tilde n=(\sum_{i=1}^kn_i-\sum_{i=1}^kn^2_i/ \sum_{i=1}^kn_i)/(k-1)$,
	\beqrs
	E\left\{\sum_{i=1}^kn_i\left(y_i-\bar y\right)^2\right\}&=&(k-1)\tilde n\tau^2+(k-1)\sigma_{\rm pop}^2.
	\eeqrs
	Thus, $E({\rm MSB}_{\rm MA})=\tilde n\tau^2+\sigma_{\rm pop}^2$.
	
	As for $E({\rm MSW}_{\rm MA})=\sigma_{\rm pop}^2$, it is derived directly by the fact that $E(n_i\hat{\sigma}^2_{y_i})=\sigma_{\rm pop}^2$.
\end{proof}

With Lemma \ref{lemm1}, $E({\rm MSB}_{\rm MA}-{\rm MSW}_{\rm MA})=\tilde n\tau^2$, and $E\{{\rm MSB}_{\rm MA}+(\tilde n-1){\rm MSW}_{\rm MA}\}=\tilde n(\tau^2+\sigma_{\rm pop}^2)$. Thus, ${\rm ICC}_{\rm MA}=\tau^2/(\tau^2+\sigma_{\rm pop}^2)$ can be estimated by $({\rm MSB}_{\rm MA}-{\rm MSW}_{\rm MA})/ \{{\rm MSB}_{\rm MA}+(\tilde n-1){\rm MSW}_{\rm MA}\}$. Truncating the negative value to zero, the $I^2_{\rm ANOVA}$ statistic in (\ref{iqb}) can be derived.

\bmsection{The statistical model for the mean difference in meta-analysis}\label{appD}

For meta-analysis of studies with two arms, we start with modeling the individual patient data in a single study. In analogy with model (\ref{IPD}), we model the individual observations $y_{ij}^T$ and $y_{ij'}^C$ of the treatment group and the control group for the $i$th study as
\beqrs
\begin{aligned}
	y_{ij}^T&=\mu^T+\delta_i^T+\xi_{ij}^T,\quad j=1,\ldots,n_i^T,\\
	y_{ij'}^C&=\mu^C+\delta_i^C+\xi_{ij'}^C,\quad j'=1,\ldots,n_i^C,
\end{aligned}
\eeqrs
where the superscript ``T" represents the treatment group, and the superscript ``C" represents the control group. Similar to the assumptions in model (\ref{IPD}), we assume that $\delta_i^T$, $\xi_{ij'}^T$, $\delta_i^C$ and $\xi_{ij'}^C$ are independent of each other. For the random errors of different observations in the same study, it is natural to assume they are i.i.d. normal random errors with mean 0 and share a common variance $\sigma^2$. Then the true effect size for each study is routinely presented by the mean difference
\beqrs
{\rm MD}_i=(\mu^T+\delta_i^T)-(\mu^C+\delta_i^C).
\eeqrs

For each study, the observed mean difference is
\beqr\label{ymd}
y_i^T-y_i^C=(\mu^T-\mu^C)+(\delta_i^T-\delta_i^C)+(\frac{\sum_{j=1}^{n_i^T}\xi_{ij}}{n_i^T}-\frac{\sum_{j'=1}^{n_i^C}\xi_{ij}}{n_i^C}),
\eeqr
where $y_i^T=\sum_{j=1}^{n^T}\xi_{ij}/n_i^T$, and $y_i^C=\sum_{j=1}^{n^C}\xi_{ij}/n_i^C$. Further, let $y_i=y_i^T-y_i^C$, $\mu=\mu^T-\mu^C$, $\delta_i=\delta_i^T-\delta_i^C$, and $\epsilon_i=\sum_{j=1}^{n_i^T}\xi_{ij}/n_i^T-\sum_{j'=1}^{n_i^C}\xi_{ij}/n_i^C$. Regardless of the dependence between $\delta_i^T$ and $\delta_i^C$, we simply assume that $\delta_i$ are i.i.d. normal random variables with mean 0 and variance $\tau^2\geq 0$, where $\tau^2$ measures the magnitude of the heterogeneity between the studies. Then model (\ref{ymd}) reduces to
\beqr\label{remmd}
y_i=\mu+\delta_i+\epsilon_i,
\eeqr
where $\delta_i\stackrel{\text{i.i.d.}}{\sim} N(0,\tau^2)$ and $\epsilon_i\stackrel{\text{ind}}{\sim} N(0,(1/n_i^T+1/n_i^C)\sigma^2)$. We note that model (\ref{remmd}) has the same form as in (\ref{REM}), except for the variance of $\epsilon_i$.

To estimate ${\rm ICC}_{\rm MA}$ for the mean difference based on ANOVA, we apply the results for the single-arm studies directly. Letting $n_i=1/(1/n_i^T+1/n_i^C)$, Lemma \ref{lemm1} in \hyperref[appB]{Appendix B} also holds that
\beqrs
E({\rm MSB}_{\rm MA})&=&\tilde n\tau^2+\sigma^2,\\
E({\rm MSW}_{\rm MA})&=&\sigma^2.
\eeqrs
Together with the notation of $\tilde n$, the $I^2_{\rm ANOVA}$ statistic in (\ref{iqb}) can be derived.

The $I^2_{\rm A}$ statistic in (\ref{iq12}) is derived similar to that for the mean. Furthermore, similar to the expression in formula (\ref{unequalv}), the $I_{\rm A}^2$ statistic can also be applied and well interpreted when the population variances differ across the studies.

\bmsection{The statistical model for the  standardized mean difference in meta-analysis}\label{appE}
For the standardized mean difference, we model the individual observations $y_{ij}^T$ and $y_{ij'}^C$ of the treatment group and the control group for the $i$th study as
\beqrs
\begin{aligned}
	y_{ij}^T&=\sigma_i(\mu^T+\delta_i^T+\xi_{ij}^T),\quad j=1,\ldots,n_i^T,\\
	y_{ij'}^C&=\sigma_i(\mu^C+\delta_i^C+\xi_{ij'}^C),\quad j'=1,\ldots,n_i^C,
\end{aligned}
\eeqrs
where the superscript ``T" represents the treatment group, and the superscript ``C" represents the control group. Similar to the assumptions in model (\ref{IPD}), we assume that $\delta_i^T$, $\xi_{ij'}^T$, $\delta_i^C$ and $\xi_{ij'}^C$ are independent of each other. In (ipdsmd), $\xi_{ij'}^T$ and $\xi_{ij'}^C$ are assumed to be i.i.d. normal random errors with mean 0 and variance 1. Then with different values of $\sigma_i$, the population variances for different studies are $\sigma_i^2$, respectively. To eliminate the influence of the scale, SMDs are considered to represent the effect sizes, which is defined by
\beqrs
{\rm SMD}_i&=&\{(\sigma_i\mu^T+\sigma_i\delta_i^T)-(\sigma_i\mu^C+\sigma_i\delta_i^C)\}/\sigma_i=(\mu^T+\delta_i^T)-(\mu^C+\delta_i^C).
\eeqrs

For each study, SMD$_i$ is estimated by
\beqr\label{ymd2}
&&\frac{y_i^T-y_i^C}{\hat\sigma_i}=\frac{\sigma_i}{\hat\sigma_i}\{(\mu^T-\mu^C)+(\delta_i^T-\delta_i^C)+(\frac{\sum_{j=1}^{n_i^T}\xi_{ij}}{n_i^T}-\frac{\sum_{j'=1}^{n_i^C}\xi_{ij}}{n_i^C})\},
\eeqr
where $\hat\sigma_i$ is an estimate for $\sigma_i$, $y_i^T=\sum_{j=1}^{n^T}\xi_{ij}/n_i^T$, and $y_i^C=\sum_{j=1}^{n^C}\xi_{ij}/n_i^C$. For simplicity, we assume that $\sigma_i$ can be accurately estimated and thus $\sigma_i/\hat\sigma_i=1$. Further, let $y_i=y_i^T-y_i^C$, $\mu=\mu^T-\mu^C$, $\delta_i=\delta_i^T-\delta_i^C$, and $\epsilon_i=\sum_{j=1}^{n_i^T}\xi_{ij}/n_i^T-\sum_{j'=1}^{n_i^C}\xi_{ij}/n_i^C$. Regardless of the dependence between $\delta_i^T$ and $\delta_i^C$, we simply assume that $\delta_i$ are i.i.d. normal random variables with mean 0 and variance $\tau^2\geq 0$, where $\tau^2$ measures the magnitude of the heterogeneity between the studies. Then model (\ref{ymd2}) reduces to
\beqr\label{remsmd}
y_i=\mu+\delta_i+\epsilon_i,
\eeqr
where $\delta_i\stackrel{\text{i.i.d.}}{\sim} N(0,\tau^2)$ and $\epsilon_i\stackrel{\text{ind}}{\sim}N(0,1/n_i^T+1/n_i^C)$. We note that model (\ref{remsmd})
has the same form as in (\ref{REM}), except for the variance of $\epsilon_i$. 

To estimate ${\rm ICC}_{\rm MA}$ for the standardized mean difference based on ANOVA, we also apply the results for the single-arm studies directly. Letting $n_i=1/(1/n_i^T+1/n_i^C)$, Lemma \ref{lemm1} in \hyperref[appB]{Appendix B} also holds that
\beqrs
E({\rm MSB}_{\rm MA})&=&\tilde n\tau^2+1.
\eeqrs
Together with the notation of $\tilde n$ and ${\rm MSB}_{\rm MA}=1$, the $I^2_{\rm ANOVA}$ statistic in (\ref{iqb}) can be derived.

\bmsection{Comparison between the $I^2$ and $I^2_{\rm A}$ statistics}\label{appa}
\begin{proof}[Proof of (a)]
Firstly, $(k-1)(\tilde n-1)$ in (\ref{iq1}) can be expressed as
\beqrs
\sum_{i=1}^kn_i-\frac{\sum_{i=1}^kn_i^2}{\sum_{i=1}^kn_i}-(k-1)&=&\frac{(\sum_{i=1}^kn_i)^2-\sum_{i=1}^kn_i^2-(k-1)\sum_{i=1}^kn_i}{\sum_{i=1}^kn_i}\\
&=&\frac{\{\sum_{i=1}^k(n_i-1)+k\}^2-\sum_{i=1}^k\{(n_i-1)+1\}^2-(k-1)\{\sum_{i=1}^k(n_i-1)+k\}}{\sum_{i=1}^kn_i}\\
&=&\frac{\{\sum_{i=1}^k(n_i-1)\}^2-\sum_{i=1}^k(n_i-1)^2+(k-1)\sum_{i=1}^k(n_i-1)}{\sum_{i=1}^kn_i}.
\eeqrs
Since the sample sizes $n_i\ge 1$ for all the studies, we have $\{\sum_{i=1}^k(n_i-1)\}^2-\sum_{i=1}^k(n_i-1)^2\ge 0$. Noting also that $k\ge 2$, it further yields that $\sum_{i=1}^kn_i-\sum_{i=1}^kn_i^2/\sum_{i=1}^kn_i-(k-1)\ge 0$, and the equality holds only when $n_i=1$ for all studies.
\end{proof}

\begin{proof}[Proof of (b)]
		For the balanced design, the weights are given by $w_i=n/\sigma_{\rm pop}^2$. Hence,
		\beqrs
	\frac{Q}{(k-1)(\tilde n-1)}&=&\frac{\sum_{i=1}^kw_i(y_i-\sum_{i=1}^kw_iy_i/\sum_{i=1}^kw_i)^2}{(k-1)(n-1)}\\
	&=&\frac{n}{n-1}\cdot\frac{1}{\sigma_{\rm pop}^2}\cdot\frac{\sum_{i=1}^k(y_i-\sum_{i=1}^ky_i/k)^2}{k-1}.
		\eeqrs
		As $n\to\infty$, $y_i$ converges in distribution to $N(\mu,\tau^2)$. Therefore, $\sum_{i=1}^k(y_i-\sum_{i=1}^ky_i/k)^2/\tau^2$ converges in distribution to $\chi^2(k-1)$. Along with the fact that ${\rm Var}\{\sum_{i=1}^k(y_i-\sum_{i=1}^ky_i/k)^2/(k-1)\}\to 0$ as $k\to \infty$, it follows that $Q/\{(k-1)(n-1)\}$ converges in probability to $\tau^2/\sigma_{\rm pop}^2$ as $k\to \infty$ and $n\to\infty$.
		
		Similarly, for any fixed $k$, it can be drived that $Q/(k-1)=(n/\sigma_{\rm pop}^2)\sum_{i=1}^k(y_i-\sum_{i=1}^ky_i/k)^2/(k-1)=O(n)$.
\end{proof}

\begin{proof}[Proof of (c)]
		When all other sample sizes are fixed, we have
		\beqrs
		\frac{\partial{\{(k-1)(\tilde n-1)\}}}{\partial n_i}=\frac{\partial{\left\{n_i+\sum_{j\neq i}n_j-(n_i^2+\sum_{j\ne i}n_j^2)/(n_i+\sum_{j\ne i}n_j)\right\}}}{\partial n_i}=\frac{2\sum_{j\ne i}n_j^2}{n_i^2+2n_i\sum_{j\ne i}n_j+\sum_{j\ne i}n_j^2}>0.
		\eeqrs
		This shows that $(k-1)(\tilde n-1)$ is an increasing function of $n_i$ given that all other sample sizes are fixed.
		
	For the unbalanced design, by noting that $w_i=n_i/\sigma_{\rm pop}^2$, we have
\beqrs
Q=\frac{1}{\sigma_{\rm pop}^2}\sum_{i=1}^kn_i\left(y_i-\frac{\sum_{i=1}^kn_iy_i}{\sum_{i=1}^kn_i}\right)^2=\frac{1}{\sigma_{\rm pop}^2}\left\{\sum_{i=1}^kn_iy_i^2-\frac{(\sum_{i=1}^kn_iy_i)^2}{\sum_{i=1}^kn_i} \right\}.
\eeqrs
As $n_i\to\infty$, $y_i$ converges in distribution to $N(\mu,\tau^2)$. With $y_i\sim N(\mu,\tau^2)$, we have $\{(n_i-\sum_{i=1}^kn_i^2)/\sum_{i=1}^kn_i\}^{-1}E\left\{\sum_{i=1}^kn_iy_i^2-(\sum_{i=1}^kn_iy_i)^2/\sum_{i=1}^kn_i\right\}$ converges to $\tau^2$. Additionally, when $k\to\infty$ and $n_i$ are of the same order, $\{(n_i-\sum_{i=1}^kn_i^2)/\sum_{i=1}^kn_i\}^{-2}{\rm Var}\left\{\sum_{i=1}^kn_iy_i^2-(\sum_{i=1}^kn_iy_i)^2/\sum_{i=1}^kn_i\right\}\to 0$. Thus, as $n\to \infty$ and $k\to\infty$, $Q/\{(k-1)(\tilde n-1)\}\to\tau^2/\sigma_{pop}^2$.
\end{proof}

\begin{proof}[Proof of (d)]
By (\ref{iq1}), we have
\beqrs
I_{\rm A}^2&=&\max\left\{\frac{Q-(k-1)}{Q+(k-1)(\tilde n-1)},0\right\}\\
         &=&\max\left\{\frac{
	\left\{Q-(k-1)\right\}/\left(\sum_{i=1}^kw_i-\sum_{i=1}^kw_i^2/\sum_{i=1}^kw_i\right)}{\left\{Q-(k-1)\right\}/\left(\sum_{i=1}^kw_i-\sum_{i=1}^kw_i^2/\sum_{i=1}^kw_i\right)+\left\{(k-1)\tilde n\right\}/\left(\sum_{i=1}^kw_i-\sum_{i=1}^kw_i^2/\sum_{i=1}^kw_i\right)},0\right\},
\eeqrs
where $\hat\tau^2={\mbox{max}}\{\{Q-(k-1)\}/(\sum_{i=1}^kw_i-\sum_{i=1}^kw_i^2/\sum_{i=1}^kw_i),0\}$ and $\tilde\sigma_y^2=(k-1)/(\sum_{i=1}^kw_i-\sum_{i=1}^kw_i^2/\sum_{i=1}^kw_i)$. Note that the above equality holds for any $Q$ value, and in case $Q<k-1$, the both sides of the last equation are zero and it still holds. 

\end{proof}

%In addition to formula (\ref{iq1}), $I_{\rm A}^2$ can be expressed through a simple transformation as 
%	\beqr\label{iqt}
%	I_{\rm A}^2=\max\left\{\frac{\hat\tau^2}{\hat\tau^2+\tilde\sigma_{\rm pop}^2},0\right\},
%	\eeqr
%	where $\hat\tau^2$ is the Dersimonian-Laird estimator, as discussed in Section \ref{sec1}, and $\tilde\sigma_{\rm pop}^2=(\sum_{i=1}^kn_i-\sum_{i=1}^kn_i^2/\sum_{i=1}^kn_i)/(\sum_{i=1}^kw_i-\sum_{i=1}^kw_i^2/\sum_{i=1}^kw_i)$. The heterogeneity statistic $I_{\rm A}^2$ in (\ref{iq1}) is derived under the assumption of a common population variance, where $n_i\sigma_{y_i}^2=\sigma_{\rm pop}^2$. When this assumption holds,  $\tilde\sigma_{\rm pop}^2$ simplifies to the common population variance $\sigma_{\rm pop}^2$. If the common population variance assumption does not hold, the $I_{\rm A}^2$ statistic can also be calculated, with $\tilde\sigma_{\rm pop}^2$ in (\ref{iqt}) interpreted as a typical population variance. Similar to the property of $\tilde\sigma_y^2$, \cite{bohning2017some} $\tilde\sigma_{\rm pop}^2$ is asymptotically equivalent to the harmonic mean of the sample sizes $n_i$ multiplied by the harmonic mean of the within-study variances $\sigma_{y_i}^2$.
\begin{figure}[htp!]
	\begin{center}
		\begin{tabular}{cc}
			\psfig{figure=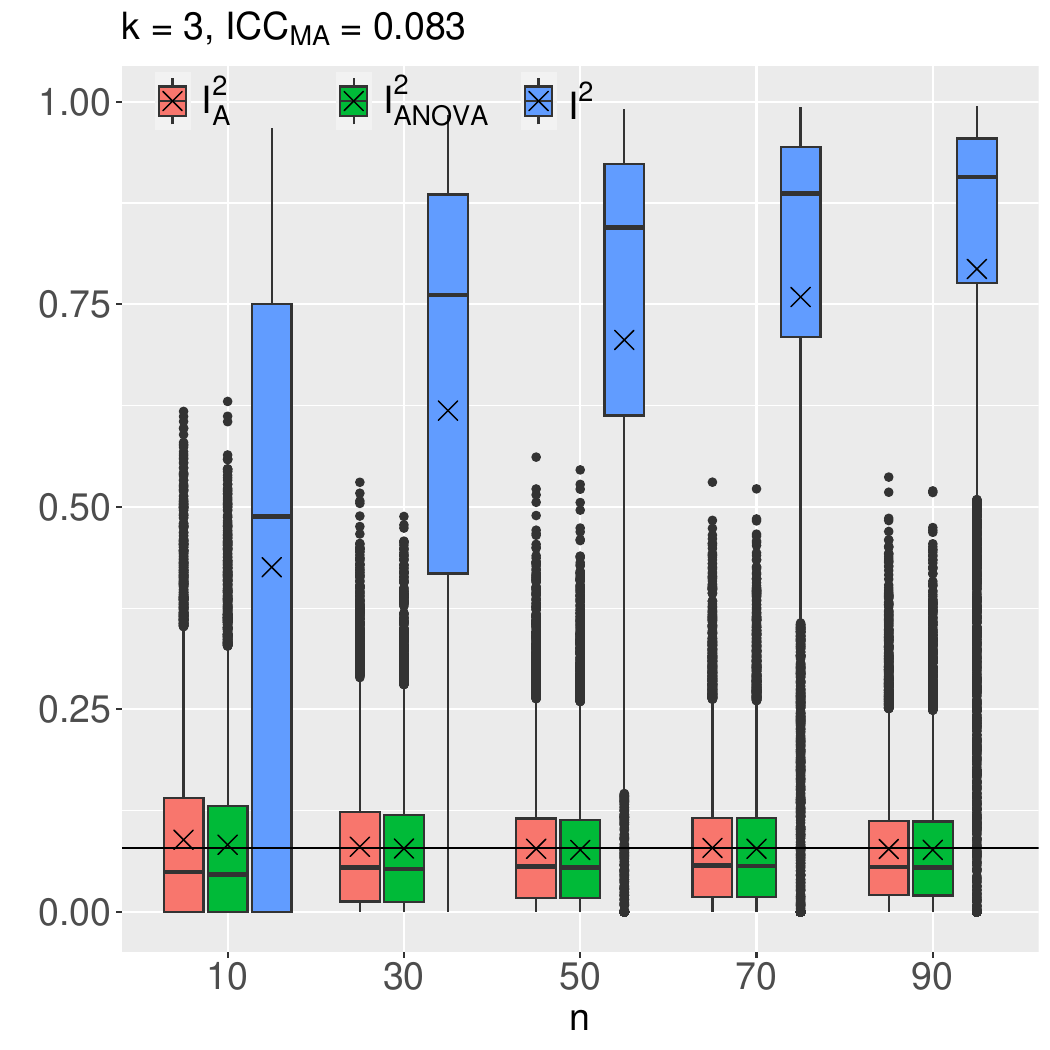,width=3.3in,angle=0}&
			\psfig{figure=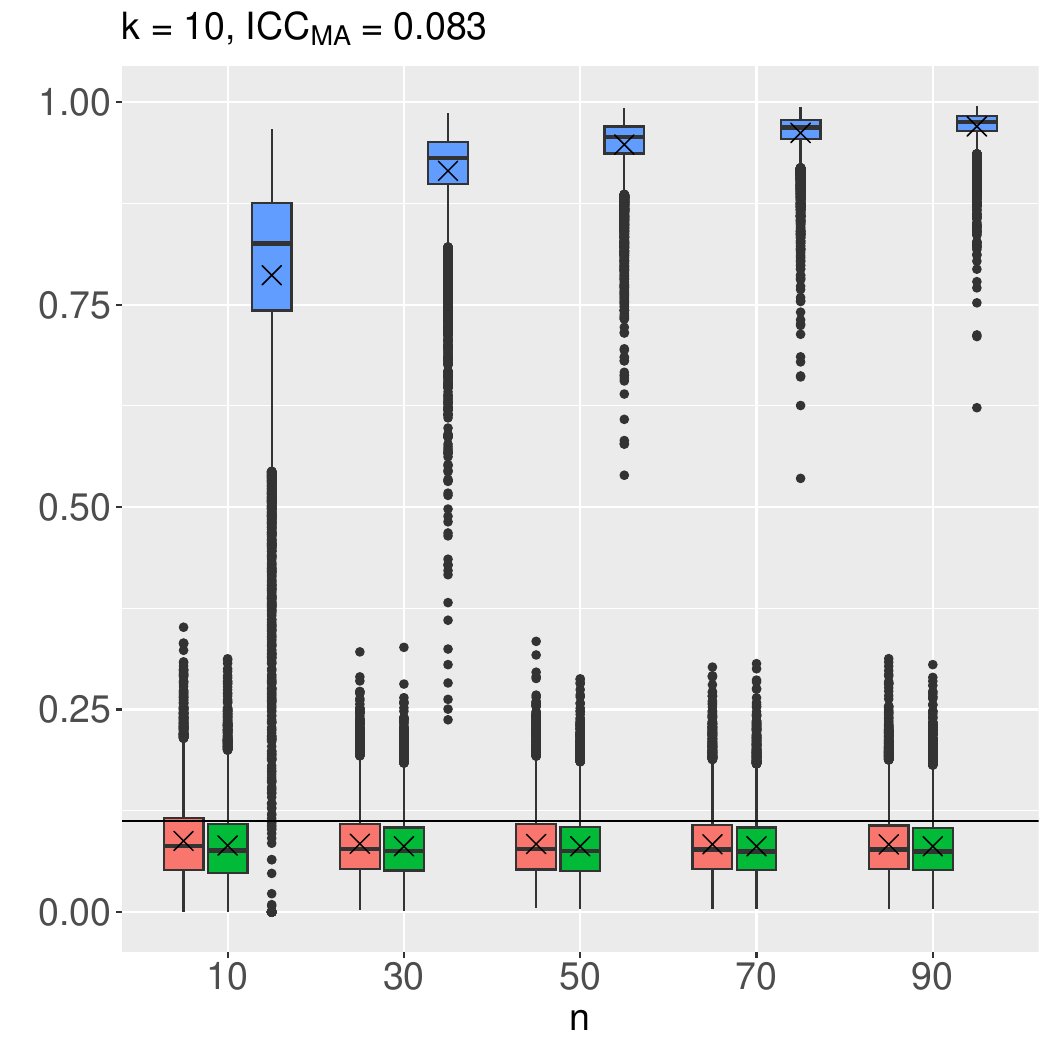,width=3.3in,angle=0}\\
			\psfig{figure=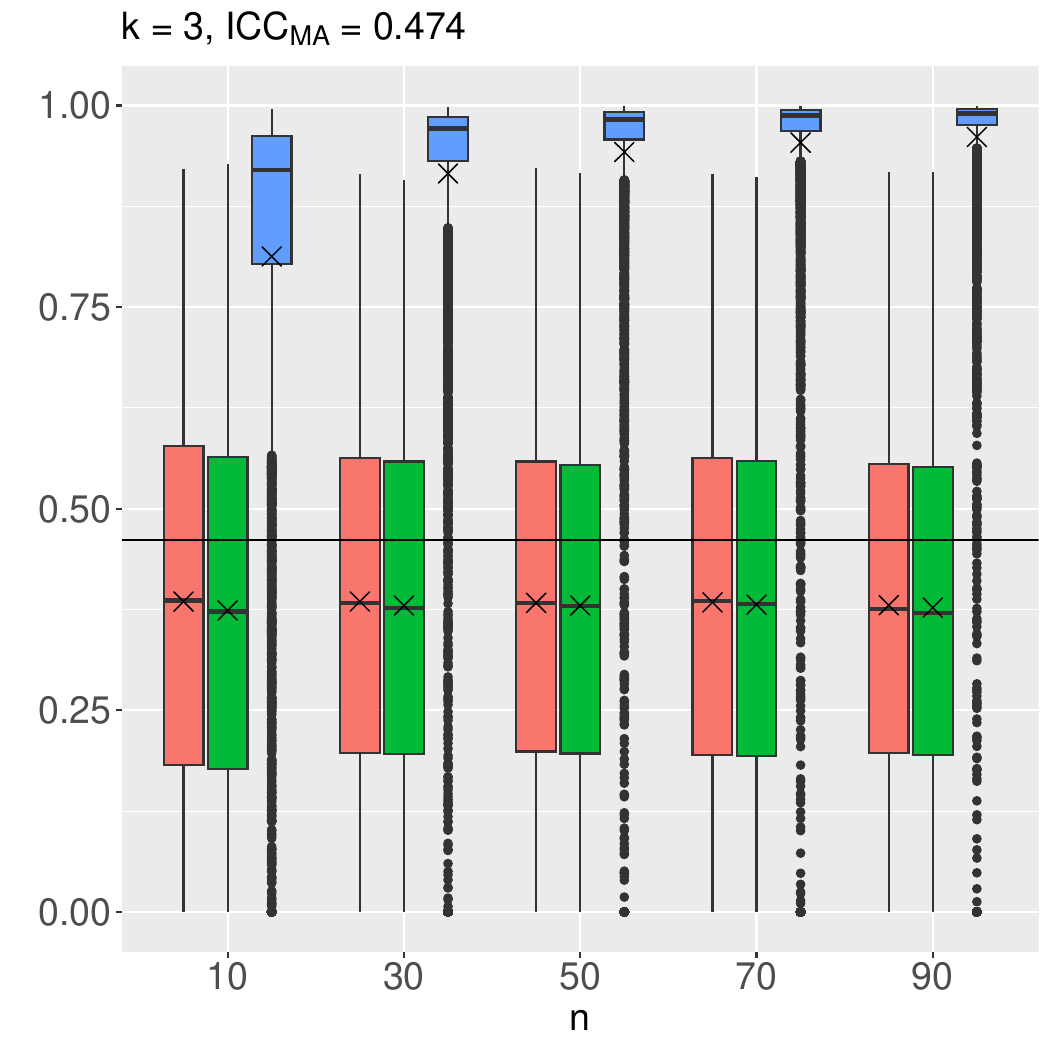,width=3.3in,angle=0}&
			\psfig{figure=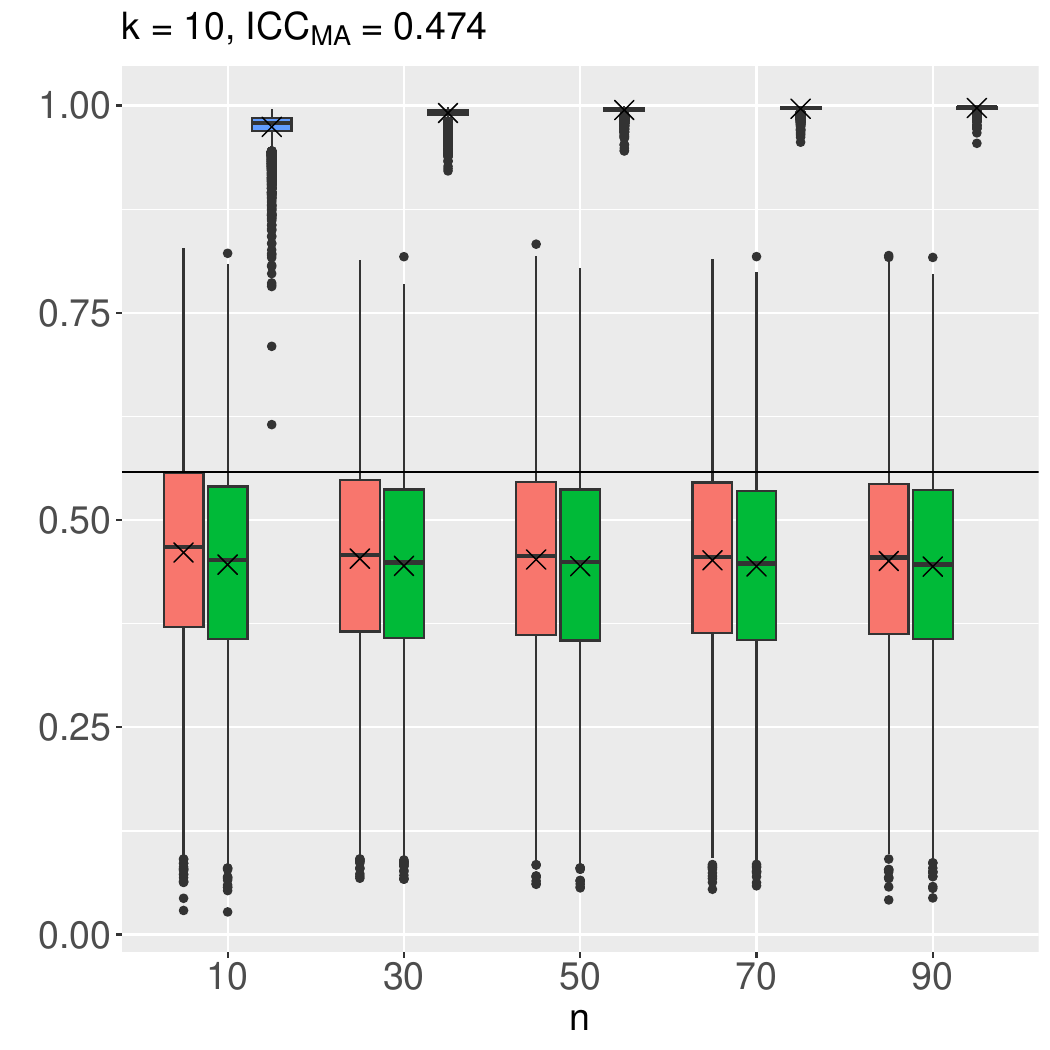,width=3.3in,angle=0}
		\end{tabular}{\caption{Boxplots of the three statistics for the mean with 10,000 repetitions. The red boxes represent the $I^2_{\rm A}$ statistic, the green boxes represent the $I^2_{\rm ANOVA}$ statistic, and the blue boxes represent the $I^2$ statistic. The crosses on each box are the mean values of the 10000 repetitions. The solid lines stand for the absolute heterogeneity ${\rm ICC}_{\rm MA}$ with $\sigma_{\rm pop}^2=100$.}\label{app}}
	\end{center}
\end{figure}

\end{document}